\documentclass[envcountsame]{llncs}
\usepackage{fullpage}
\usepackage{microtype}
\usepackage{amsmath,amsfonts,amssymb,stmaryrd,dsfont,mathtools}
\usepackage{booktabs}
\usepackage{subfig}
\usepackage{enumerate}
\usepackage{tikz}
\usetikzlibrary{arrows,decorations,snakes,calc,shapes}

\usepackage[colorlinks=true, citecolor=blue,linkcolor=red,urlcolor=black]{hyperref}

\usepackage{avg}

\let\doendproof\endproof
\renewcommand\endproof{~\hfill\qed\doendproof}

\title{Average-energy games\thanks{Work partially supported by
    European project CASSTING (FP7-ICT-601148) and ERC project EQualIS
    (StG-308087). Mickael Randour is an F.R.S.-FNRS Postdoctoral Researcher.}}
\author{Patricia~Bouyer\inst{1} \and Nicolas Markey\inst{1} \and Mickael
  Randour\inst{2} \and
Kim~G.~Larsen\inst{3} \and Simon Laursen\inst{3}}

\institute{
LSV -- CNRS \& ENS Cachan, France
\and Computer Science Department, Universit\'e libre de Bruxelles (ULB), Belgium
\and Aalborg University, Denmark
}

\begin{document}

\maketitle

\begin{abstract}
Two-player quantitative zero-sum games provide a natural framework to
synthesize controllers with performance guarantees for reactive systems within
an uncontrollable environment. Classical settings include mean-payoff games,
where the objective is to optimize the long-run average gain per action, and
energy games, where the system has to avoid running out of energy.

We study \textit{average-energy} games, where the goal is to optimize the
long-run average of the accumulated energy. We show that this objective arises
naturally in several applications, and that it yields interesting connections with
previous concepts in the literature. We~prove that deciding the winner in such
games is in \NP $\cap$ \coNP and at least as hard as solving mean-payoff
games, and we establish that memoryless strategies suffice to~win. We~also
consider the case where the system has to minimize the average-energy
\textit{while} maintaining the accumulated energy within predefined bounds at
all times: this~corresponds to operating with a finite-capacity storage for
energy. We give results for one-player and two-player games, and establish
complexity bounds and memory requirements.

\end{abstract}

\section{Introduction}
\label{sec:intro}
\paragraph{Quantitative games.} Game-theoretic formulations are a standard
tool for the synthesis of provably-correct controllers for reactive
systems~\cite{GTW02}. We consider two-player (system vs. environment)
turn-based games played on finite graphs. Vertices of the graph are called
\textit{states} and partitioned into states of player~1 and states of
player~2. The game is played by moving a pebble from state to state, along
\textit{edges} in the graph, and starting from a given initial state. Whenever
the pebble is on a state belonging to player~$i$, player~$i$ decides where to
move the pebble next, according to his \textit{strategy}. The infinite path
followed by the pebble is called a \textit{play}: it~represents one possible
behavior of the system. A~\textit{winning objective} encodes acceptable
behaviors of the system and can be seen as a set of winning plays. The goal of
player~1 is to ensure that the outcome of the game will be a winning play,
whatever the strategy played by his adversary.

To reason about resource constraints and the performance of strategies, \textit{quantitative games} have been considered in the literature. See for example~\cite{emsoft2003-CAHS,BCHJ09,Ran13}, or~\cite{Ran14} for an overview. Those games are played on \textit{weighted} graphs, where edges are fitted with integer weights modeling rewards or costs. The performance of a play is evaluated via a \textit{payoff function} that maps it to the numerical domain. The objective of player~1 is then to ensure a sufficient payoff with regard to a given threshold value. Seminal classes of quantitative games include mean-payoff~($\MPG$), total-payoff~($\TPG$) and energy games~($\EG$). In~$\MPG$ games~\cite{EM79,ZP96,ipl68(3)-Jur}, player~1 has to optimize his long-run average gain per edge taken whereas, in $\TPG$ games~\cite{mfcs2004-GZ,GS09}, player~1 has to optimize his long-run sum of weights. Energy games~\cite{emsoft2003-CAHS,BFLMS08,JLR13} model safety-like properties: the goal is to ensure that the running sum of weights never drops below zero and/or that it never exceeds a given upper bound $U \in \mathbb{N}$. All three classes share common properties. First, $\MPG$~games, $\TPG$ games, and $\EG$ games with only a lower bound ($\EGL$) are memoryless determined (given an initial state, either player~1 has a strategy to win, or player~2 has one, and in both cases no memory is required to win). Second, deciding the winner for those games is in \NP $\cap$ \coNP and no polynomial algorithm is known despite many efforts (e.g.,~\cite{BCDGR11,Chatterjee201525}). Energy games with both lower and upper bounds~($\EGLU$) are more complex: they are \EXPTIME-complete and winning requires memory in general~\cite{BFLMS08}.

While those classes are well-known, it is sometimes necessary to go beyond them to accurately model practical applications. For example, multi-dimensional games and conjunctions with a parity objective model trade-offs between different quantitative aspects~\cite{CD10,CRR14,VCDHRR15}. Similarly, window objectives address the need for strategies ensuring good quantitative behaviors within reasonable time frames~\cite{Chatterjee201525}.

\paragraph{Average-energy games.} We study the \textit{average-energy} ($\AE$) payoff function: in $\AE$ games, the goal of player~1 is to optimize the \textit{long-run average accumulated energy} over a play. We introduce this objective to formalize the specification desired in a practical application~\cite{CJLRR09}, which we detail in the following as a motivating example. Interestingly, it turns out that this payoff first appeared long ago~\cite{TV87}, but it was not subject to a systematic study until very recently: see related work for more discussion.

In addition to being meaningful w.r.t.~practical applications, $\AE$ games also have theoretical interest. In~\cite{CP13}, Chatterjee and Prabhu define the \textit{average debit-sum level} objective, which can be seen as a variation of the \textit{average-energy} where the accumulated energy is taken to be zero in any point where it is actually positive (hence, it focuses on the average debt). They use the corresponding games to compute the values of quantitative timed simulation functions. In particular, they provide a pseudo-polynomial-time algorithm to solve those games, but the complexity of deciding the winner as well as the memory requirements are open. Here, we solve those questions for the very similar average-energy objective. 

\paragraph{Motivating example.} Our example is a simplified version of the
industrial application studied by Cassez et~al.~in~\cite{CJLRR09}. Consider a
machine that consumes oil, stored in a connected accumulator. We want to
synthesize an appropriate controller to operate the oil pump that fills the
accumulator, and by the effect of pressure, that releases oil from the accumulator into the machine with a (time-varying) rate according to desired production. In order to ensure safety, the oil level in the accumulator should be maintained at all times between a minimal and a maximal level. This part of the specification can be encoded as an energy objective with both lower and upper bounds~($\EGLU$). At the same time, the more oil (thus pressure) in the accumulator, the faster the whole apparatus wears out. Hence, an ideal controller should minimize the average level of oil in the long run. This desire can be formalized through the average-energy payoff~($\AE$). Overall, the specification is thus to minimize the average-energy under the strong energy constraints: we denote the corresponding objective by~$\AELU$.

\renewcommand{\arraystretch}{1.1}
\begin{table*}[t]\centering
\small
\scalebox{1}{\begin{tabular}{cccccc}\toprule
Game objective & \textbf{1-player} && \textbf{2-player}&& \textbf{memory}  \\
 \midrule

$\MPG$  & \PTIME~\cite{Kar78} && \NP $\cap$ \coNP~\cite{ZP96} && memoryless~\cite{EM79} \\
$\TPG$  & \PTIME~\cite{FV97} && \NP $\cap$ \coNP~\cite{GS09}  && memoryless~\cite{mfcs2004-GZ} \\
$\EGL$  & \PTIME \cite{BFLMS08} && \NP $\cap$ \coNP \cite{emsoft2003-CAHS,BFLMS08} && memoryless~\cite{emsoft2003-CAHS} \\
$\EGLU$ & \PSPACE-complete \cite{FJ13} && \EXPTIME-complete \cite{BFLMS08} && pseudo-polynomial \\
 \midrule

$\AEG$  & \PTIME && \NP $\cap$ \coNP && memoryless \\
$\AELU$, polynomial~$U$ & \PTIME && \NP $\cap$ \coNP && polynomial \\
$\AELU$, arbitrary~$U$ & \PSPACE-complete && \EXPTIME-complete && pseudo-polynomial \\
$\AEL$  & \PSPACE-easy~/~\NP-hard && \textit{open}~/~\EXPTIME-hard && \textit{open} ($\geq$ pseudo-p.)\\
\bottomrule
\end{tabular}}
\vspace{1mm}

\caption{Complexity of deciding the winner and memory requirements for quantitative games: $\MP$~stands for mean-payoff, $\TP$ for total-payoff, $\EGL$ (resp.~$\EGLU$) for lower-bounded (resp.~lower- and upper-bounded) energy, $\AE$ for average-energy, and $\AEL$ (resp.~$\AELU$) for average-energy under a lower bound (resp.~and upper bound $U \in \mathbb{N}$) on the energy. Results without reference are proved in this paper.}
\label{tab:results}
\vspace{-4mm}
\end{table*}

\paragraph{Contributions.} Our main results are summarized in Table~\ref{tab:results}.

A) We establish that the average-energy objective can be seen as a \textit{refinement} of total-payoff, in the same sense as total-payoff is seen as a refinement of mean-payoff~\cite{GS09}: it allows to distinguish strategies yielding identical mean-payoff and total-payoff. 

B) We show that deciding the winner in two-player $\AE$ games is in \NP $\cap$ \coNP whereas it is in \PTIME for one-player games. In both cases, memoryless strategies suffice (Thm.~\ref{thm:ae_two_memoryless}). Those complexity bounds match the state-of-the-art for $\MP$ and $\TP$ games~\cite{ZP96,ipl68(3)-Jur,GS09,BCDGR11}. Furthermore we prove that $\AE$ games are at least as hard as mean-payoff games (Thm.~\ref{thm:mp_to_ae}). Therefore, the \NP $\cap$ \coNP-membership can be considered optimal w.r.t.~our knowledge of $\MP$ games.
Technically, the crux of our approach is as follows. First, we show that memoryless strategies suffice in one-player $\AE$ games (Thm.~\ref{thm:aeg_memoryless}): this requires to prove important properties of the $\AE$ payoff as classical sufficient criteria for memoryless determinacy present in the literature fail to apply directly. Second, we establish a polynomial-time algorithm for the one-player case: it exploits the structure of winning strategies and mixes graph techniques with local linear program solving (Thm.~\ref{thm:ae_onePlayer_PTIME}). Finally, we lift memoryless determinacy to the two-player case using results by Gimbert and Zielonka~\cite{GZ05} and obtain the \NP~$\cap$~\coNP-membership as a corollary (Thm.~\ref{thm:ae_npinter}). 
\enlargethispage{3mm}

C) We establish an \EXPTIME algorithm to solve two-player $\AE$ games with lower- and upper-bounded energy ($\AELU$) with an arbitrary upper bound $U \in \mathbb{N}$ (Thm.~\ref{thm:aelu_reduc}). It relies on a reduction of the $\AELU$ game to a pseudo-polynomially larger $\AE$ game where the energy constraints are encoded in the graph structure. Applying straightforwardly the $\AE$ algorithm on this game would only give us \NEXPTIME~$\cap$~\coNEXPTIME-membership, hence we avoid this blowup by further reducing the problem to a particular $\MP$ game and applying a pseudo-polynomial algorithm, with some care to ensure that overall the algorithm only requires pseudo-polynomial time in the original $\AELU$ game. Since the simpler $\EGLU$ games (i.e., $\AELU$ with a trivial $\AE$ constraint) are already \EXPTIME-hard~\cite{BFLMS08}, the $\EXPTIME$-membership result is optimal. We also prove that pseudo-polynomial memory is both sufficient and in general necessary to win in $\AELU$ games, for both players (Thm.~\ref{thm:aelu_memory}). We show that one-player $\AELU$ games are $\PSPACE$-complete via the on-the-fly construction of a witness path based on the aforementioned reduction, answering a question left open in~\cite{DBLP:journals/corr/BouyerMRLL15}. For polynomial (in the size of the game graph) values of the upper bound~$U$\,---\,or~if it is given in unary\,---\,the~complexity of the two-player (resp.~one-player) $\AELU$ problem collapses to \NP $\cap$ \coNP (resp.~\PTIME) with the same approach, and polynomial memory suffices for both players.

D) We provide partial answers for the $\AEL$ objective\,---\,$\AE$ under a lower
bound constraint on energy but no upper bound. We~show \PSPACE-membership for the one-player case (Thm.~\ref{thm:ael_one_complexity}), by reducing the problem to an $\AELU$ game
with a sufficiently large upper bound. That is, we prove that if the player
can win for the $\AEL$ objective, then he can do so without ever increasing
its energy above a well-chosen bound. We also prove the $\AEL$ problem to be
at least \NP-hard in one-player games (Thm.~\ref{thm:ael_one_complexity}) and \EXPTIME-hard in two-player games
(Lem.~\ref{lem:ael_exp_hard}) via reductions from the subset-sum problem and
countdown games respectively. Finally, we show that memory is required for
both players in two-player $\AEL$ games (Lem.~\ref{lem:ael_memory}), and that
pseudo-polynomial memory is both sufficient and necessary in the one-player
case (Thm.~\ref{thm:ael_one_memory}). The~decidability status of two-player $\AEL$ games remains open as we only provide a correct but incomplete incremental algorithm (Lem.~\ref{lem:ael_semi}). We conjecture that the two-player $\AEL$ problem is decidable and sketch a potential approach to solve it. We highlight the key remaining questions and discuss some connections with related models that are known to be difficult. 

Observe that in many applications, the energy must be stocked in a finite-capacity storage for which an upper bound is provided. Hence, the model of choice in this case is~$\AELU$.

\paragraph{Related work.} This paper extends previous work presented in a conference~\cite{DBLP:journals/corr/BouyerMRLL15}: it gives a full
presentation of the technical details, along with additional results and improved complexities.

The \textit{average-energy} payoff\,---\,Eq.~\eqref{eq:ae}\,---\,appeared in a paper by
Thuijsman and Vrieze in the late eighties~\cite{TV87}, under the name
\textit{total-reward}. This definition is different from the classical
\textit{total-payoff}\,---\,see~Sect.~\ref{sec:prelim}\,---\,commonly studied in the formal
methods community (see for example~\cite{mfcs2004-GZ,GS09}), which, despite
that, has been referred in many papers as either total-payoff or total-reward
equivalently. We will see in this paper that both definitions are
\textit{indeed} different and exhibit different behaviors.

Maybe due to this confusion, the payoff of Eq.~\eqref{eq:ae}\,---\,which we call
\textit{average-energy} thus avoiding misunderstandings\,---\,was not studied
extensively until recently. Nothing was known about memoryless determinacy and
complexity of deciding the winner. Independently to our work, Boros et
al.~recently studied the same payoff (under the name \textit{total-payoff}).
In~\cite{BEGM15}, they study Markov decision processes and stochastic
games with the payoff of Eq.~\eqref{eq:ae} and solve both questions. Their
results overlap with ours for $\AE$ games (Table~\ref{tab:results}). Let us
first mention that our results were obtained independently. Second, and
\textit{most importantly}, our approach and \textit{techniques are different},
and we believe our take on the problem yields some interest for our community.
Indeed, the algorithm of Boros et~al.~entirely relies on linear programming in
the one-player case, and resorts to approximation by discounted games in the
two-player one. Our techniques are arguably more constructive and based on
inherent properties of the payoff. In that sense, it is closer to what is
usually deemed important in our field. For example, we provide an extensive
comparison with classical payoffs. We base our proof of memoryless determinacy
on \textit{operational understanding} of the $\AE$ which is crucial in order
to formalize proper specifications. Our technique then benefits from seminal
works~\cite{GZ05} to bypass the reduction to discounted games and obtain a
direct proof, thanks to our more constructive approach. Lastly,
while~\cite{BEGM15} considers the $\AEG$ problem in the stochastic
context, we focus on the deterministic one but consider multi-criteria
extensions by adding bounds on the energy ($\AELU$ and $\AEL$ games). Those
extensions are \textit{completely new}, exhibit theoretical interest and are
adequate for practical applications in constrained energy systems, as
witnessed by the case study of~\cite{CJLRR09}.

Recent work of Br\'azdil et~al.~\cite{BKKN14} considers the optimization of a payoff under energy constraint. They study mean-payoff in consumption systems, i.e., simplified one-player energy games where all edges consume energy but some states can atomically produce a reload of the energy up to the allowed capacity.

\section{Preliminaries}
\label{sec:prelim}
\paragraph{Graph games.} 
We consider turn-based games played on graphs between two players denoted by
$\pI$ and~$\pII$. A~\emph{game} is a tuple $\Game =
(S_1, S_2, \trans, \weg)$ where 
(i)~$S_1$~and $S_2$ are disjoint finite sets of
\textit{states} belonging to $\pI$ and~$\pII$, with $S = S_1
\uplus S_2$, 
(ii)~$\trans \subseteq S \times S$ is a finite set of \textit{edges} such that for all $s \in S$, there exists $s' \in S$ such that $(s, s') \in \trans$ (i.e., no deadlock), and 
(iii)~$\weg\colon \trans \to \bbZ$ is an integer \textit{weight
  function}. 
Given edge $(s_{1}, s_{2}) \in \trans$, we write $\weg(s_{1},
s_{2})$ as a shortcut for $\weg((s_{1}, s_{2}))$. We denote by~$\largestW$ the
largest absolute weight assigned by function $\weg$. A~game is
called $1$-player if~$\states_{1} = \emptyset$ or $\states_{2} = \emptyset$.

A \emph{play} from an initial state $\initState \in \states$ is an
infinite sequence $\play = s_0 s_1 \ldots s_n \ldots$ such that $s_0 =
\initState$ and for all $i \ge 0$ we have $(s_i,s_{i+1}) \in \trans$.
The (finite) \emph{prefix} of $\play$ up to position $n$ gives the
sequence $\play(n) = s_0 s_1 \ldots s_n$, the first (resp. last)
element $s_0$ (resp. $s_n$) is denoted $\first(\play(n))$
(resp. $\last(\play(n))$). The set of all plays in $\Game$ is denoted
by $\plays(\Game)$ and the set of all prefixes is denoted by
$\prefs(\Game)$. We say that a prefix $\prefix \in \prefs(\Game)$
belongs to $\player{i}$, $i \in \{1,2\}$, if $\last(\prefix) \in
S_i$. The set of prefixes that belong to $\player{i}$ is denoted by
$\prefs_i (\Game)$. The classical concatenation between prefixes
(resp. prefix and play) is denoted by the $\cdot$ operator. The length
of a non-empty prefix $\prefix = s_{0}\ldots{}s_{n}$ is defined as the
number of edges and denoted by $\vert\prefix\vert = n$.

\paragraph{Payoffs of plays.} Given a play $\play = s_0 s_1 \ldots s_n \ldots$ we define
\begin{itemize}
\item its \textit{energy level} at position $n$ as
	$\EL(\play(n)) = \sum_{i = 0}^{n-1} w(s_i,s_{i+1})$;

\item its \textit{mean-payoff} as
	$\MPsup(\play) = 
	\limsup_{n \to \infty} \frac{1}{n}  \sum_{i = 0}^{n-1} w(s_i,s_{i+1}) 
	= \limsup_{n \to \infty} \frac{1}{n} \EL(\play(n))$;
	
\item its \textit{total-payoff} as
	$\TPsup(\play) = 
	\limsup_{n \to \infty} \sum_{i = 0}^{n-1} w(s_i,s_{i+1}) 
	= \limsup_{n \to \infty} \EL(\play(n))$;
\item and its \textit{average-energy} as
\vspace{-1mm}
\begin{equation}
\label{eq:ae}
  	\AEsup(\play) = 
	\limsup_{n \to \infty} \frac{1}{n} \sum_{i = 1}^{n} 
	\left ( \sum_{j = 0}^{i-1} w(s_j,s_{j+1}) \right )
	= \limsup_{n \to \infty} \frac{1}{n} \sum_{i = 1}^{n}  \EL(\play(i)).
\end{equation}
\end{itemize}
\vspace{-1mm}

We will sometimes consider those measures defined with $\liminf$ instead of $\limsup$, in which case we write $\MPinf$, $\TPinf$ and $\AEinf$ respectively. Finally, we also consider those measures over prefixes: we naturally define
them by dropping the $\limsup_{n \rightarrow \infty}$ operator and taking $n =
\vert\prefix\vert$ for a prefix $\prefix \in \prefs(\Game)$. In this case, we
simply write $\MP(\prefix)$, $\TP(\prefix)$ and $\AE(\prefix)$ to denote the
fact that we consider \textit{finite} sequences.

\paragraph{Strategies.} A~\emph{strategy} for $\player{i}$, $i \in \{1,2\}$, is a function $\St_i\colon \prefs_i(\Game)
\to S$ such that for all $\prefix \in \prefs_i(\Game)$ we have $(\last(\prefix),
\St_i(\prefix)) \in \trans$.
A~strategy~$\St_{i}$ for~$\player{i}$ is \textit{finite-memory} 
if it can be
encoded by a deterministic Moore machine $(M,m_0,\alpha_u,\alpha_n)$ where $M$
is a finite set of states (the memory of the strategy), $m_0 \in M$ is the
initial memory state, $\alpha_u \colon M \times S \to M$ is an update
function, and $\alpha_n \colon M \times S_{i} \to S$ is the next-action
function. If the game is in $s \in S_{i}$ and $m \in M$ is the current memory
value, then the strategy chooses $s' = \alpha_n(m,s)$ as the next state of the
game. When the game leaves a state $s \in S$, the memory is updated to
$\alpha_u(m,s)$. Formally, $(M, m_0, \alpha_u,
  \alpha_n) $ defines the strategy $\St_{i}$ such that
$\St_{i}(\rho\cdot s) = \alpha_n(\hat{\alpha}_u(m_0, \rho), s)$ for all $\rho
\in S^*$ and $s \in S_{i}$, where $\hat{\alpha}_u$ extends $\alpha_u$ to
sequences of states as expected. A strategy is \emph{memoryless} if $\vert
M\vert = 1$, i.e., it does not depend on the history but only on the current
state of the game. We denote by $\strats_{i}(\Game)$, the sets of strategies for player $\player{i}$. We drop $\Game$ when the context is clear.

A play $\play = s_0 s_1 \ldots$ is \emph{consistent} with a
strategy $\St_i$ of $\player{i}$ if, for all $n \ge 0$ where $\last(\play(n))
\in S_i$, we have $\St_i(\play(n)) = s_{n+1}$. Given an initial state
$\initState \in \states$ and strategies $\St_{1}$ and $\St_{2}$ for the two
players, we denote by $\out(\initState, \St_1,\St_2 )$ the unique play that
starts in $\initState$ and is consistent with both $\St_1$ and $\St_2$. When
fixing the strategy of only $\player{i}$, we denote the set of
consistent outcomes by $\outs(\initState, \St_i)$.

\paragraph{Objectives.} 
An~\emph{objective} in $\Game$ is a set $\mathcal{W} \subseteq \plays(\Game)$ that is declared
winning for $\pI$. Given a game $\Game$, an initial state~$\initState$, and an
objective~$\mathcal{W}$, a~strategy $\St_1 \in \strats_{1}$ is winning for $\playerOne$
if for all strategy $\St_2 \in \strats_{2}$, we have that $\out(\initState,
\St_1,\St_2) \in \mathcal{W}$. Symmetrically, a~strategy $\St_2 \in \strats_{2}$ is
winning for $\playerTwo$ if for all strategy $\St_1 \in \strats_{1}$, we~have
that $\out(\initState, \St_1,\St_2) \not\in \mathcal{W}$. That is, we consider \textit{zero-sum} games.

We consider the following objectives and combinations of those objectives.
\begin{itemize}
	\item Given an initial energy level $\initCredit \in \bbN$, the \textbf{lower-bounded energy} ($\EGL$) objective
		$\LBound(\initCredit) = \{ \play \in \plays(G)$ $\mid \forall\, n
		 \ge 0,\ \initCredit + \EL(\play(n)) \geq 0 \}$
	requires non-negative energy at all times. 
	
	\item Given an upper bound $U \in \bbN$ and an initial energy level
	$\initCredit \in \bbN$, the \textbf{lower- and upper-bounded energy} ($\EGLU$) objective
			$\LUBound(U, \initCredit) = \{ \play \in \plays(\Game) \mid 
			\forall\, n \ge 0,\ \initCredit + \EL(\play(n)) \in [0,U] \}$
		requires that the energy always remains non-negative and below the 
		upper bound $U$ along a play.

	\item Given a threshold $t \in \bbQ$, the \textbf{mean-payoff} ($\MPG$) objective 
		$\MeanPayOff(t) = \{ \play \in \plays(\Game) \mid \MPsup(\play) \le t \}$
	requires that the mean-payoff is at most~$t$. 
	
	\item Given a threshold $t \in \bbZ$, the \textbf{total-payoff} ($\TPG$) objective 
		$\TotalPayOff(t) = \{ \play \in \plays(\Game) \mid \TPsup(\play) \le t \}$
	requires that the total-payoff is at most~$t$.

	\item Given a threshold $t \in \bbQ$, the \textbf{average-energy} ($\AEG$) objective
		$\AvgEnergyLevel(t) = \{ \play \in \plays(\Game) \mid \AEsup(\play) \le t \}$
	requires that the average-energy is at most~$t$.
\end{itemize}

For the $\MPG$, $\TPG$ and $\AEG$ objectives, note that
$\playerOne$ aims to \textit{minimize} the payoff value while $\playerTwo$
tries to maximize~it. The reversed convention is also often used in the
literature but both are equivalent. For our motivating example, 
seeing $\playerOne$ as a minimizer is more natural. Note that we define the objectives using the $\limsup$ variants
of~$\MPG$, $\TPG$ and~$\AEG$, but similar results are obtained for the $\liminf$ variants.

\paragraph{Decision problem.} 
Given a game $\Game$, an initial state $\initState \in \states$, and an
objective $\mathcal{W} \subseteq \plays(\Game)$ as defined above, the associated
\textit{decision problem} is to decide if $\pI$ has a winning strategy for
this objective.

We recall classical results in Table~\ref{tab:results}. Memoryless strategies suffice for both players for $\EGL$~\cite{emsoft2003-CAHS,BFLMS08}, $\MPG$~\cite{EM79} and
$\TPG$~\cite{FV97,mfcs2004-GZ} objectives. Since all associated
problems can be solved in polynomial time for 1-player games, it follows that
the 2-player decision problem is in $\NP\cap\coNP$ for those three objectives~\cite{BFLMS08,ZP96,GS09}. For the
$\EGLU$ objective, memory is in general needed and the associated
decision problem is \EXPTIME-complete~\cite{BFLMS08} (\PSPACE-complete for one-player games~\cite{FJ13}).

\paragraph{Game values.} Given a game with an objective $\mathcal{W} \in \{\MeanPayOff, \TotalPayOff, \AvgEnergyLevel\}$ and an initial state $\initState$, we refer to the \textit{value} from $\initState$ as $v = \inf \{t \in \mathbb{Q} \mid \exists\, \sigma_{1} \in \Sigma_{1},\, \outs(\initState, \St_1) \subseteq \mathcal{W}(t)\}$. For~both $\MPG$ and~$\TPG$ objectives, it~is known that the value can be achieved by an optimal memoryless strategy; for the $\AEG$ objective it follows from our results (Thm.~\ref{thm:ae_two_memoryless}).

\section{Average-Energy}
\label{sec:average_games}
In this section, we consider the problem of ensuring a \textit{sufficiently low} average-energy.
\begin{bproblem}[$\AEG$] 
  Given a game~$\Game$, an initial state $\initState$, and a threshold~$t \in \mathbb{Q}$, decide if $\pI$ has a winning strategy $\St_1 \in \strats_{1}$ for the objective
  $\AvgEnergyLevel(t)$.
\end{bproblem}

We first compare the $\AEG$ objective with traditional quantitative objectives and study how they can be connected (Sect.~\ref{subsec:ae_relation}). Then we want to establish that in $\AEG$ games, memoryless strategies are always sufficient to play optimally, for \textit{both} players. Interestingly, this result cannot be obtained by straightforward application of many well-known sufficient criteria for memoryless determinacy existing in the literature. We thus introduce some technical lemmas that highlight the inherent features of the $\AEG$ payoff function (Sect.~\ref{subsec:ae_tech}) and permit to prove the result for \textit{one-player} $\AEG$ games (Sect.~\ref{subsec:ae_one}). We then prove that one-player $\AEG$ games can be solved in polynomial-time via an algorithm combining graph analysis techniques with linear programming. Finally, we consider the \textit{two-player} case (Sect.~\ref{subsec:ae_two}). Applying a result by Gimbert and Zielonka~\cite{GZ05}, combined with our results on the one-player case, we derive memoryless determinacy of two-player $\AEG$ games. This also induces $\NP \cap \coNP$-membership of the $\AEG$ problem by the $\PTIME$ algorithm of Sect.~\ref{subsec:ae_one}. We conclude by proving that $\AEG$ games are at least as hard as $\MPG$ games, hence indicating that the $\NP \cap \coNP$ upper bound is essentially optimal with regard to our current knowledge of $\MPG$ games (whose membership to $\PTIME$ is a long-standing open problem~\cite{ZP96,ipl68(3)-Jur,BCDGR11,Chatterjee201525}).

\subsection{Relation with classical objectives}
\label{subsec:ae_relation}

Several links between $\EGL$, $\MPG$
and $\TPG$ objectives can be established. Intuitively, $\pI$ can only ensure a lower bound on energy if he can prevent $\pII$ from enforcing strictly-negative cycles (otherwise the initial energy is eventually exhausted). This is the case if and only if $\pI$ can ensure a non-negative mean-payoff in $\Game$ (here, he wants to maximize the $\MP$), and if this is the case, $\pI$ can prevent the running sum of weights from ever going too far below zero along a play, hence granting a lower bound on total-payoff. 
We introduce the sign-reversed game $\Game'$ in the next lemma, which is consistent with our view of $\pI$ as a minimizer with regard to payoffs (as discussed in Sect.~\ref{sec:prelim}).
\begin{lemma}
\label{lem:simple_objectives_equivalences}
Let $\Game = (S_1, S_2, \trans, \weg)$ be a game and $\initState \in \states$
be the initial state. The following assertions are equivalent.
\begin{enumerate}[A.]
\item\label{assert:energy} There exists $\initCredit \in \mathbb{N}$ such that
  $\playerOne$ has a (memoryless) winning strategy for objective
  $\LBound(\initCredit)$.
\item\label{assert:MP} Player $\playerOne$ has a (memoryless) winning strategy
  for objective $\MeanPayOff(0)$ in the game $\Game'$ defined by reversing the
  sign of the weight function, i.e., for all $(s_{1}, s_{2}) \in \trans$,
  $\weg'(s_{1}, s_{2}) = -\weg(s_{1}, s_{2})$.
\item\label{assert:TP} Player $\playerOne$ has a (memoryless) winning strategy
  for objective $\TotalPayOff(t)$, with $t = 2\cdot(\vert\states\vert -
  1)\cdot \largestW$, in the game $\Game'$ defined by reversing the sign of
  the weight function.
\item\label{assert:TPinfinite} There exists $t \in \mathbb{Z}$ such that
  $\playerOne$ has a (memoryless) winning strategy for objective
  $\TotalPayOff(t)$, in the game $\Game'$ defined by reversing the sign of the
  weight function.
\end{enumerate}
\end{lemma}
\begin{proof}
  Proof of $\ref{assert:energy} \Leftrightarrow \ref{assert:MP}$ is given
  in~\cite[Proposition 12]{BFLMS08}. Proof~of $\ref{assert:MP}
  \Leftrightarrow \ref{assert:TP} \Leftrightarrow \ref{assert:TPinfinite}$
 is in~\cite[Lem.~1]{Chatterjee201525}.
\end{proof}

The $\TPG$ objective is sometimes seen as a \textit{refinement} of
$\MPG$ for the case where~$\pI$\,---\,as~a minimizer\,---\,can ensure~$\MP$ equal
to zero but not lower, i.e., the $\MPG$ game has value zero~\cite{GS09}. Indeed, one may use
the $\TP$ to further discriminate between strategies that guarantee
$\MP$~zero. In~the same philosophy, the~average-energy can
help in distinguishing strategies that yield identical total-payoffs. See Fig.~\ref{fig:ae_refines}. The~$\AE$ values in both examples can be computed easily using the upcoming technical lemmas (Sect.~\ref{subsec:ae_tech}).

\begin{figure}[htb]
        \centering
\subfloat{\scalebox{1}{\begin{tikzpicture}[->,>=stealth',shorten >=1pt,auto,node
    distance=2.5cm,bend angle=45, scale=0.42, yscale=.9,font=\normalsize,inner sep=.5mm]
    \everymath{\scriptstyle}
    \tikzstyle{p1}=[draw,circle,text centered,minimum size=5mm,text width=4mm]
    \tikzstyle{p2}=[draw,rectangle,text centered,minimum size=5mm,text width=4mm]
    \node[p1]  (0)  at (0, 0) {};
    \node[p1]  (1) at (3, 0) {};
    \node[p1]  (2) at (6, 1.6) {};
    \node[p1]  (3) at (6, -1.6) {};
    \node[p1]  (4) at (9, 0) {};
    
    \coordinate[shift={(-5mm,0mm)}] (init) at (0.west);
    \path
    (0) edge node[above] {$1$} (1)
    (init) edge (0);
	\draw[->,>=latex] (1) to[out=60,in=180] node[above] {$2$} (2);
	\draw[->,>=latex] (2) to[out=0,in=120] node[above] {$2$} (4);
	\draw[->,>=latex] (4) to[out=240,in=0] node[below] {$-2$} (3);
	\draw[->,>=latex] (3) to[out=180,in=300] node[below] {$-2$} (1);
      \end{tikzpicture}} }
      \hspace{2cm}
      \subfloat{\scalebox{1}{\begin{tikzpicture}[->,>=stealth',shorten >=1pt,auto,node
    distance=2.5cm,bend angle=45, scale=0.42, yscale=.9, font=\normalsize,inner sep=.5mm]
    \everymath{\scriptstyle}
    \tikzstyle{p1}=[draw,circle,text centered,minimum size=5mm,text width=4mm]
    \tikzstyle{p2}=[draw,rectangle,text centered,minimum size=5mm,text width=4mm]
    \node[p1]  (0)  at (0, 0) {};
    \node[p1]  (1) at (3, 0) {};
    \node[p1]  (2) at (6, 1.6) {};
    \node[p1]  (3) at (6, -1.6) {};
    \node[p1]  (2b) at (9, 1.6) {};
    \node[p1]  (3b) at (9, -1.6) {};
    \node[p1]  (4) at (12, 0) {};
    
    \coordinate[shift={(-5mm,0mm)}] (init) at (0.west);
    \path
    (2) edge node[above] {$2$} (2b)
    (3b) edge node[below] {$-2$} (3)
    (0) edge node[above] {$1$} (1)
    (init) edge (0);
	\draw[->,>=latex] (1) to[out=60,in=180] node[above] {$2$} (2);
	\draw[->,>=latex] (2b) to[out=0,in=120] node[above] {$0$} (4);
	\draw[->,>=latex] (4) to[out=240,in=0] node[below] {$0$} (3b);
	\draw[->,>=latex] (3) to[out=180,in=300] node[below] {$-2$} (1);
      \end{tikzpicture}} }

\subfloat[Play $\play_{1}$ sees energy levels $(1, 3, 5, 3)^{\omega}$.]{\scalebox{1}{\begin{tikzpicture}[scale=.8]
      \everymath{\scriptstyle}
		\def \xscale {0.4}
		\def \yscale {0.4}
		\def \xmax {12}
		\def \ymax {6}
			
		\def\xy#1#2{({#1*\xscale},{#2*\yscale})}
                \path[use as bounding box] (-1,-.8) -- ({\xmax*\xscale+1},{\ymax*\yscale + 0.6});
	    \draw[->] (-0.2,0) -- ({\xmax*\xscale+0.5},0) node[right] {\small Step};
	    \draw[->] (0,-0.2) -- (0,{\ymax*\yscale + 0.2}) node[above] {\small Energy};
 		
        \foreach \y in {0,2,...,\ymax} {	
			\draw (-0.1,{\y*\yscale}) node[anchor=east] {$\y$} ;
		}
		
        \foreach \x in {0,2,...,\xmax} {
			\draw ({\x*\xscale},-0.1) node[anchor=north] {$\x$} ;
		}

        \foreach \x in {0,...,\xmax} {
            \draw ({\x*\xscale},0) -- ({\x*\xscale},-1.5pt);
        }

        \foreach \y in {0,...,\ymax} {
            \draw (0, {\y*\yscale}) -- (-1.5pt, {\y*\yscale});
        }
		
		\def\mean{3}
		\draw[thick,dashed,color=black]
			\xy{0}{\mean} -- \xy{\xmax}{\mean};
		
		\draw \xy{\xmax}{\mean} node[anchor=west] {$ \mathit{AE} = 3$ };

		\draw[very thick,-,color=black]
					\xy{0}{1} --
					\xy{2}{5} --
					\xy{4}{1} --
					\xy{6}{5} --
					\xy{8}{1} --
					\xy{10}{5} --
					\xy{12}{1}  ;							 
	\end{tikzpicture}}}
	\hspace{1.5cm}
\subfloat[Play $\play_{2}$ sees energy levels $(1, 3, 5, 5, 5, 3)^{\omega}$.]{\scalebox{1}{
\begin{tikzpicture}[scale=.8]
					
      \everymath{\scriptstyle}
		\def \xscale {0.4}
		\def \yscale {0.4}
		\def \xmax {12}
		\def \ymax {6}
			
		\def\xy#1#2{({#1*\xscale},{#2*\yscale})}
               \path[use as bounding box] (-1,-.8) -- ({\xmax*\xscale+1.5},{\ymax*\yscale + 0.6});
	    \draw[->] (-0.2,0) -- ({\xmax*\xscale+0.5},0) node[right] {\small Step};
	    \draw[->] (0,-0.2) -- (0,{\ymax*\yscale + 0.2}) node[above]
              {\small Energy};
 		
        \foreach \y in {0,2,...,\ymax} {	
			\draw (-0.1,{\y*\yscale}) node[anchor=east] {$\y$} ;
		}
		
        \foreach \x in {0,2,...,\xmax} {
			\draw ({\x*\xscale},-0.1) node[anchor=north] {$\x$} ;
		}

        \foreach \x in {0,...,\xmax} {
            \draw ({\x*\xscale},0) -- ({\x*\xscale},-1.5pt);
        }

        \foreach \y in {0,...,\ymax} {
            \draw (0, {\y*\yscale}) -- (-1.5pt, {\y*\yscale});
        }
 	
		\def\mean{3.67}
		\draw[thick,dashed,color=black]
			\xy{0}{\mean} -- \xy{\xmax}{\mean};
		
		\draw \xy{\xmax}{\mean} node[anchor=west] {$ \mathit{AE} = 11/3$ };
		
		\draw[very thick,-,color=black]
					\xy{0}{1} --
					\xy{2}{5} --
					\xy{4}{5} --
					\xy{6}{1} --
					\xy{8}{5} --
					\xy{10}{5} --
					\xy{12}{1} ;		
	\end{tikzpicture}
}}
	\caption{Both plays have identical mean-payoff and total-payoff: $\MPsup(\play_{1}) = \MPinf(\play_{1}) = \MPsup(\play_{2}) = \MPinf(\play_{2}) = 0$, $\TPsup(\play_{1}) = \TPsup(\play_{2}) = 5$, and $\TPinf(\play_{1}) = \TPinf(\play_{2}) = 1$. But play $\play_{1}$ has a lower average-energy: $\AEsup(\play_{1}) = \AEinf(\play_{1}) =  3 < \AEsup(\play_{2}) = \AEinf(\play_{2}) = 11/3$.}
	\label{fig:ae_refines}
\end{figure}

In these examples, the average-energy is clearly comprised between the infimum and supremum total-payoffs. This remains true for any play. 

\begin{lemma}
\label{lem:AEbetweenTP}
For any play $\play \in \plays(\Game)$, we have that $\AEinf(\play), \AEsup(\play) \in \left[
  \TPinf(\play), \TPsup(\play) \right] \subseteq \mathbb{R} \cup \{-\infty,
\infty\}$.
\end{lemma}

\begin{proof}
  Consider a play $\play \in \plays(\Game)$. By definition of the total-payoff
  and thanks to weights taking integer values, we have that there exists some
  index $m \in \mathbb{N}_{0}$ such that, for all $n \geq m$, $\EL(\play(n))
  \in \left[ \TPinf(\play), \TPsup(\play) \right]$. By definition, the
  average-energy $\AEsup$ (resp.~$\AEinf$) measures the supremum (resp.~infinimum) limit of the averages of those partial
  sums, hence it holds that $\AEinf(\play), \AEsup(\play) \in \left[ \TPinf(\play),
    \TPsup(\play) \right]$.
\end{proof}

In particular, \textit{if the mean-payoff value from a state is not zero}, its total-payoff value is infinite and the following lemma holds.

\begin{lemma}
\label{lem:MPInfToAE}
  Let $\Game = (S_1, S_2, \trans, \weg)$ be a game and $\initState \in
  \states$ be the initial state.
\begin{enumerate}
\item\label{prop:MPneg} If there exists $t < 0$ such that $\playerOne$ has a (memoryless)
  winning strategy for $\MeanPayOff(t)$, then $\playerOne$ has a memoryless
  strategy that is winning for $\AvgEnergyLevel(t')$ for all $t' \in
  \mathbb{Q}$, i.e., this strategy ensures that any consistent outcome $\play$
  is such that $\AEinf(\play) = \AEsup(\play) = -\infty$.
\item\label{prop:MPpos} If $\playerOne$ has no (memoryless) winning strategy for
  $\MeanPayOff(0)$, then, for any $t' \in \mathbb{Q}$, $\playerOne$ has no
  winning strategy for $\AvgEnergyLevel(t')$. In particular, $\playerTwo$ has
  a memoryless strategy ensuring that any consistent outcome $\play$ is such
  that $\AEinf(\play) = \AEsup(\play) = \infty$.
\end{enumerate}
\end{lemma}

\begin{proof}
  Consider the first implication. Assume $\playerOne$ has a memoryless
  strategy $\St_{1}$ ensuring that all consistent outcomes $\play \in
  \outs(\initState, \St_1)$ are such that $\MPsup(\play) < 0$. For any such
  outcome, it is guaranteed that all simple cycles have a strictly negative
  energy level. Thus, we have that $\TPsup(\play) = -\infty$, and by
  Lem.~\ref{lem:AEbetweenTP}, it implies that $\AEsup(\play) = -\infty$, as
  claimed. Since $\AEinf(\play) \leq \AEsup(\play)$ by definition, the property holds.

Now consider the second implication. Assume there exists no winning strategy
for $\playerOne$ for the mean-payoff objective. By equivalence \ref{assert:MP}
$\Leftrightarrow$ \ref{assert:TPinfinite} of
Lem.~\ref{lem:simple_objectives_equivalences}, and memoryless determinacy of
total-payoff games (see for example~\cite{mfcs2004-GZ}), it follows that
$\playerTwo$ has a memoryless strategy $\St_{2}$ ensuring that all consistent
outcomes $\play \in \outs(\initState, \St_2)$ are such that $\TPinf(\play) =
\infty$. By~Lem.~\ref{lem:AEbetweenTP}, this induces the claim.
\end{proof}

\subsection{Useful properties of the average-energy}
\label{subsec:ae_tech}

In this subsection, we will first review some classical criteria that usually prove sufficient to deduce memoryless determinacy in quantitative games and discuss why they cannot be applied straight out of the box to the average-energy payoff. We will then prove two useful properties of this payoff that will later help us to prove the desired result.

\paragraph{Classical sufficient criteria.}

We briefly discuss
traditional approaches to prove memoryless determinacy in quantitative games.
The first one is to study a variant of the infinite-duration game where
the game halts as soon as a cycle is closed and then to relate the properties
of this variant to the infinite-duration game. This technique was used in the
original proof of memoryless determinacy for mean-payoff games by Ehrenfeucht
and Mycielski~\cite{EM79}, and in a following simpler proof by Bj{\"{o}}rklund
et~al.~\cite{BSV04}. The connection between infinite-duration games and
so-called \textit{first cycle games} was recently streamlined by Aminof and
Rubin~\cite{AR14}, identifying sufficient conditions to prove that first cycle
games and their infinite-duration counterparts admit optimal memoryless
strategies for both players. Among those conditions is the need for winning
objectives to be closed under \textit{cyclic permutation} (intuitively, swapping cycles in a play should not induce a better payoff) and under
\textit{concatenation} (intuitively, concatenating two prefixes should not result in a payoff better than the best of the two prefixes). Without further assumptions, the average-energy
objective satisfies neither. Indeed, consider individual cycles represented by sequences
of \textit{weights} $\mathcal{C}_{1} = \{-1\}$, $\mathcal{C}_{2} = \{1\}$ and
$\mathcal{C}_{3} = \{1, -2\}$. We see that
$\AE(\mathcal{C}_{1}\mathcal{C}_{2}) = (-1 + 0)/2 = -1/2 <
\AE(\mathcal{C}_{2}\mathcal{C}_{1}) = (1 -0)/2 = 1/2$, hence $AE$ is not
closed under cyclic permutations. Intuitively, the order in which the weights
are seen \textit{does} matter, in contrast to most classical payoffs. For
concatenation, see that $\AE(\mathcal{C}_{3}) = 0$ while
$\AE(\mathcal{C}_{3}\mathcal{C}_{3}) = -1/2 < 0$. Here the intuition is that
the overall $\AE$ is impacted by the energy of the first cycle
which is strictly negative ($-1$). In a sense, the $\AE$ of a cycle
can only be maintained through repetition if this cycle is neutral with regard
to the total energy level, i.e., if it has energy level zero: we will formalize this
intuition in Lem.~\ref{lem:AE_repeat}.

Other criteria for memoryless determinacy or half-memoryless determinacy (i.e., holding only for one of the two players) respectively appear in works by Gimbert and Zielonka~\cite{mfcs2004-GZ} and by Kopczynski~\cite{Kop06}. They involve checking that the payoff is \textit{fairly mixing}, or \textit{concave}. Again, both are false for arbitrary sequences of weights in the case of the average-energy, for essentially the same reasons as above. Nevertheless, we will be able to prove that memoryless strategies suffice for both players using similar ideas but first taking care of the problematic cases. Intuitively, when those cases are dealt with, we will regain a payoff that satisfies the above conditions. We also obtain \textit{monotonicity} and \textit{selectivity} of the payoff function as defined in~\cite{GZ05}. 

\paragraph{Extraction of prefixes.} The following lemma describes the impact of adding a finite prefix to an infinite play. We prove that the average-energy over a play can be decomposed w.r.t.~to the energy level of any of its prefixes and the average-energy of the remaining suffix.

\begin{lemma}[\textbf{Average-energy prefix}]
\label{lem:AE_prefix}
Let $\prefix \in \prefs(\Game)$, $\play \in \plays(\Game)$. Then, $\AEsup(\prefix \cdot \play) = \EL(\prefix) + \AEsup(\play)$. The same equality holds for $\AEinf$.
\end{lemma}

\begin{proof}
  Let $\prefix = s_{0}\ldots{}s_{k} \in \prefs(\Game)$ and $\play \in
  \plays(\Game)$ be a prefix and a play over a game~$\Game$. We prove the property for $\AEsup$. By~definition and
  decomposition, we have that
\begin{align*}
\AEsup(\prefix \cdot \play) &= 
	\limsup_{n \to \infty} \dfrac{1}{n} \sum_{i = 1}^{n} 
	\EL((\prefix \cdot \play)(i))\\
&= \limsup_{n \to \infty} \left[ \dfrac{1}{n} \cdot \sum_{i = 1}^{k} 
	\EL(\prefix(i)) + \frac{1}{n} \cdot \sum_{i = k+1}^{n} 
	\EL(\prefix) + \frac{1}{n} \cdot \sum_{i = k+1}^{n} 
	\EL(\play(i-k))\right]. 
\end{align*}
For clarity, we rewrite this expression as $\AEsup(\prefix \cdot \play) =
\limsup_{n \to \infty} \left[ X_{1}(n) + X_{2}(n) + X_{3}(n)\right]$,
maintaining the same order.

Since $k$ is fixed and finite, and $\EL(\prefix(i))$ is bounded
for all $i \leq k$, we have that $\limsup_{n \to \infty} X_{1}(n) = \lim_{n \to
  \infty} X_{1}(n) = 0$. Furthermore, for $n \geq k + 1$, we rewrite the
second term as $X_{2}(n) = (n - k - 1)\cdot \EL(\prefix)/n$, and it follows
that $\limsup_{n \to \infty} X_{2}(n) = \lim_{n \to \infty} X_{2}(n) =
\EL(\prefix)$. Since both sequences $X_{1}(n)$ and $X_{2}(n)$ converge, we can
write
\begin{align*}
  \liminf_{n \to \infty} X_{1}(n) + \liminf_{n \to \infty} X_{2}(n) +
  \limsup_{n \to \infty} X_{3}(n) \leq \AEsup(\prefix \cdot \play) \leq
  \limsup_{n \to \infty} X_{1}(n) + \limsup_{n \to \infty} X_{2}(n) +
  \limsup_{n \to \infty} X_{3}(n).
\end{align*}
Hence, by a small change of variable,
\begin{align*}
  \AEsup(\prefix \cdot \play) = \EL(\prefix) + \limsup_{n \to \infty} X_{3}(n)
  = \EL(\prefix) + \limsup_{n \to \infty} \left[\frac{1}{n} \cdot \sum_{i =
      1}^{n-k-1} \EL(\play(i))\right] = \EL(\prefix) + \AEsup(\play),
\end{align*}
as, in the limit, the $(k+1)$ missing terms in the sum are negligible. The proof for $\AEinf$ is similar.
\end{proof}

\paragraph{Extraction of a best cycle.} The next lemma is crucial to prove that memoryless strategies suffice: under well-chosen conditions, one can always select a best cycle in a play\,---\,hence, there is no interest in mixing different cycles and no use for memory. It holds only for sequences of cycles \textit{that have energy level zero}: since they do not change the energy, they do not modify the $\AE$ of the following suffix of play, and one can decompose the $\AE$ as a weighted average over zero cycles.

\begin{lemma}[\textbf{Repeated zero cycles of bounded length}]
  \label{lem:AE_repeat}
  Let $\cycle_{1}, \cycle_{2}, \cycle_{3}, \ldots{}$ be an infinite sequence
  of cycles $\cycle_{i} \in \prefs(\Game)$ such that (i) $\play = \cycle_{1}
  \cdot \cycle_{2} \cdot \cycle_{3} \cdots{} \in \plays(\Game)$,\footnote{We slightly abuse the notation as we see cycles as sequences of \textit{edges}. The concatenation of cycles $\cycle_{a} = s\,s'\ldots{} s$ and $\cycle_{b} = s\,s'' \ldots{} s$ is to be understood as its natural interpretation $\cycle_{a} \cdot \cycle_{b} = s\,s'\ldots{} s\,s'' \ldots{} s$: the origin state $s$ only appears \textit{once} in the middle and not twice as it would with $\cycle_{a}$ and $\cycle_{b}$ seen as true sequences of states.}
  (ii)
  $\forall\,i \geq 1$, $\EL(\cycle_{i}) = 0$ and (iii) $\exists\, \ell \in
  \bbN*$ such that $\forall\,i \geq 1$, $\vert\cycle_{i}\vert \leq \ell$. Then
  the following properties hold.
\begin{enumerate}
\item\label{prop:weighted_average} The average-energy of $\play$ is the
  \textit{weighted average} of the average-energies of the cycles:
\begin{equation}
  \AEsup(\play) = \limsup_{k \rightarrow \infty} \left[\dfrac{\sum_{i=1}^{k} 
      \vert\cycle_{i}\vert \cdot \AE(\cycle_{i})}{\sum_{i=1}^{k}
      \vert\cycle_{i}\vert}\right].
\end{equation}
\item\label{prop:repeated_cycle} For any cycle $\cycle \in \prefs(\Game)$ such
  that $\EL(\cycle) = 0$, we have that $\AEsup(\cycle^{\omega}) =
  \AE(\cycle)$.
\item\label{prop:repeat_best_cycle} Repeating the best cycle gives the lowest
  $\AE$:
  $\inf_{i \in \bbN*} \AE(\cycle_{i}) = \inf_{i \in \bbN*}
  \AEsup((\cycle_{i})^{\omega}) \leq \AEsup(\play)$.
\end{enumerate}
Similar properties hold for $\AEinf$.
\end{lemma}

Observe that since we assume a bound $\ell \in \bbN*$ on the length of cycles,
and the game is played on a finite graph,
Point~\ref{prop:repeat_best_cycle} of Lem.~\ref{lem:AE_repeat} does actually allow to select a \textit{best} cycle:
the set of possible cycles of length at most~$\ell$ is finite and the infimum
is reached, hence can be replaced by the miminum.

\begin{proof}
We prove the three points for $\AEsup$, similar arguments can be applied for $\AEinf$.
  Consider Point~\ref{prop:weighted_average}. Let $\play =
  s_{0}^{1}\ldots{}s_{\vert\cycle_{1}\vert}^{1}
  s_{1}^{2}\ldots{}s_{\vert\cycle_{2}\vert}^{2}s_{1}^{3} \ldots{}$ where $s_{j}^{i}$
  denotes the $j$-th state of cycle $\cycle_{i}$, with $\cycle_{1} = s_{0}^{1}\ldots{}s_{\vert\cycle_{1}\vert}^{1}$ and for all $i > 1$, $\cycle_{i} = s^{i-1}_{\vert\cycle_{i-1}\vert}s_{1}^{i}\ldots{}s^{i}_{\vert\cycle_{i}\vert}$.   Essentially, $s^{i-1}_{\vert\cycle_{i-1}\vert}$ is both the last state of $\cycle_{i-1}$ and the first one of $\cycle_{i}$: it can also be seen as $s^{i}_{0}$ and we later use both notations depending on the role we consider for this state. Given index $k \in \bbN$ of
  a state $s_{k}$ in the classical formulation $\play =
  s_{0}s_{1}s_{2}\ldots{}$ such that $s_{k}$ denotes state $s_{j}^{i}$ in our
  new formulation $\play =
  s_{0}^{1}\ldots{}s_{\vert\cycle_{1}\vert}^{1}
  s_{1}^{2}\ldots{}s_{\vert\cycle_{2}\vert}^{2}s_{1}^{3} \ldots{}$, we define $c(k) = i$ and $p(k) = j$, respectively denoting
  the index of the corresponding cycle and the position of state $s_{k}$
  within this cycle. We can rewrite the definition of the average-energy
  of~$\play$ as
\begin{align}
\label{eq:repeated_cycles}
\AEsup(\play) &=
  \limsup_{n \to \infty} \left[\dfrac{1}{n} \sum_{k = 1}^{n}
  \EL(\play(k))\right] = \limsup_{n \to \infty} \left[\dfrac{1}{n}
  \left(\sum_{i = 1}^{c(n) - 1} \sum_{j = 1}^{\vert\cycle_{i}\vert}
    \EL(s_{0}^{1}\ldots{}s_{j}^{i})\:+\: \sum_{j = 1}^{p(n)}
    \EL(s_{0}^{1}\ldots{}s_{j}^{c(n)})\right)\right].
\end{align}
Now observe that since all cycles are such that $\EL(\cycle_{i}) = 0$, we have
that $\EL(s_{0}^{1}\ldots{}s_{j}^{i}) = \EL(s_{0}^{i}\ldots{}s_{j}^{i})$ for
all indices $i \in \bbN*$, $j \in \{1, \ldots, \vert\cycle_{i}\vert\}$. In
other words, the energy level in a given position only depends on the current
cycle. Hence, for all $i \in \bbN*$,
\begin{equation*}
  \sum_{j = 1}^{\vert\cycle_{i}\vert} \EL(s_{0}^{1}\ldots{}s_{j}^{i}) =
  \sum_{j = 1}^{\vert\cycle_{i}\vert} \EL(s_{0}^{i}\ldots{}s_{j}^{i}) =
    \vert\cycle_{i}\vert \cdot \AE(\cycle_{i}) 
\end{equation*}
where the second equality follows by definition of $\AE(\cycle_{i})$.
Therefore, Eq.~\eqref{eq:repeated_cycles} becomes
\begin{equation*}
  \AEsup(\play) = \limsup_{n \to \infty} \left[\dfrac{1}{n} \left(\sum_{i =
        1}^{c(n) - 1} \vert\cycle_{i}\vert \cdot \AE(\cycle_{i})\:+\: \sum_{j
        = 1}^{p(n)} \EL(s_{0}^{c(n)}\ldots{}s_{j}^{c(n)})\right)\right]. 
\end{equation*}
Recall that, by hypothesis, there exists $\ell \in \bbN*$ such that for all $i
\geq 1$, $\vert\cycle_{i}\vert \leq \ell$. Observe that the boundedness of
cycles length implies that 
  (a)~$p(n) \leq \ell$, 
  (b)~$\sum_{j = 1}^{p(n)} \EL(s_{0}^{c(n)}\ldots{}s_{j}^{c(n)})$ is bounded,
  and 
  (c)~$\sum_{i =
    1}^{c(n) - 1} \vert\cycle_{i}\vert \leq n = \sum_{i = 1}^{c(n) - 1} \vert\cycle_{i}\vert + p(n) \leq \sum_{i =
    1}^{c(n) - 1} \vert\cycle_{i}\vert + \ell$. 
Combining those three arguments, we obtain that 
\begin{equation*}
 \limsup_{n \to \infty} \left[\dfrac{\sum_{i = 1}^{c(n) - 1}
      \vert\cycle_{i}\vert \cdot \AE(\cycle_{i})}{\sum_{i = 1}^{c(n) - 1}
      \vert\cycle_{i}\vert + \ell}\right] \leq \AEsup(\play) \leq \limsup_{n \to \infty} \left[\dfrac{\sum_{i = 1}^{c(n) - 1}
      \vert\cycle_{i}\vert \cdot \AE(\cycle_{i})}{\sum_{i = 1}^{c(n) - 1}
      \vert\cycle_{i}\vert }\right] 
\end{equation*}
Hence,
\begin{equation*}
\AEsup(\play) = \limsup_{k \rightarrow \infty}
  \left[\dfrac{\sum_{i=1}^{k} \vert\cycle_{i}\vert \cdot
      \AE(\cycle_{i})}{\sum_{i=1}^{k} \vert\cycle_{i}\vert}\right]
\end{equation*}
as claimed by Point~\ref{prop:weighted_average}.

Now consider Point~\ref{prop:repeated_cycle}. For any cycle $\cycle \in
\prefs(\Game)$ such that $\EL(\cycle) = 0$, all three hypotheses~\textit{(i)},
\textit{(ii)}, and~\textit{(iii)} are clearly satisfied, with $\ell =
\vert\cycle\vert$. Hence by Point~\ref{prop:weighted_average}, we have that
\begin{equation*}
\AEsup(\cycle^{\omega}) = \limsup_{k \rightarrow \infty} \left[\dfrac{k\cdot
    \vert\cycle\vert \cdot \AE(\cycle)}{k\cdot \vert\cycle\vert}\right] =
\AE(\cycle). 
\end{equation*}

Finally, we prove Point~\ref{prop:repeat_best_cycle}. The equality
straightforwardly follows from Point~\ref{prop:repeated_cycle}. It remains
to consider the inequality. By definition of the infimum, we have that, for
all $k \geq 1$,
\begin{equation*}
\inf_{i \in \bbN*} \AE(\cycle_{i}) = \dfrac{\sum_{i=1}^{k}
  \vert\cycle_{i}\vert \cdot \inf_{i \in \bbN*}
  \AE(\cycle_{i})}{\sum_{i=1}^{k} \vert\cycle_{i}\vert} \leq
\dfrac{\sum_{i=1}^{k} \vert\cycle_{i}\vert \cdot
  \AE(\cycle_{i})}{\sum_{i=1}^{k} \vert\cycle_{i}\vert}. 
\end{equation*}
Hence by taking the limit, we obtain 
\begin{equation*}
  \inf_{i \in \bbN*} \AE(\cycle_{i}) = \limsup_{k \rightarrow \infty}
  \left[\inf_{i \in \bbN*} \AE(\cycle_{i})\right] \leq \limsup_{k \rightarrow
    \infty} \left[\dfrac{\sum_{i=1}^{k} \vert\cycle_{i}\vert \cdot
      \AE(\cycle_{i})}{\sum_{i=1}^{k} \vert\cycle_{i}\vert}\right] =
  \AEsup(\play). 
\end{equation*}
This concludes our proof.
\end{proof}

\subsection{One-player games}

\textit{We assume that the unique player is $\playerOne$}, hence that $S_{2} = \emptyset$. The proofs are similar for the case where all states belong to $\playerTwo$ (i.e., $S_{1} = \emptyset$). Similarly, we present our results for the $\AEsup$ variant, but they carry over to the $\AEinf$ one. Actually, since we show that we can restrict ourselves to \textit{memoryless} strategies, all consistent outcomes will be periodic and thus both variants will be equal over those outcomes. 

\paragraph{Memoryless determinacy.} Intuitively, we use Lem.~\ref{lem:AE_prefix} and Lem.~\ref{lem:AE_repeat} to transform any arbitrary path into a simple \textit{lasso path}, repeating a unique simple cycle, yielding an $\AE$ at least as good, thus proving that any threshold achievable with memory can also be achieved without it.

\label{subsec:ae_one}
\begin{theorem}
\label{thm:aeg_memoryless}
Memoryless strategies are sufficient to win one-player $\AEG$ games.
\end{theorem}

\begin{proof}
  As a preliminary step, we check whether the graph contains a reachable
  strictly negative cycle, e.g., using the Bellman-Ford algorithm in $\mathcal{O}(\vert S\vert \cdot \vert T \vert)$-time. If so, then $\playerOne$ can ensure a strictly negative mean-payoff, and by Point~\ref{prop:MPneg} of Lem.~\ref{lem:MPInfToAE}, a memoryless strategy exists to make the average-energy be~$-\infty$: such a strategy consists in reaching and repeating the negative simple cycle forever.
  
  Now, assume that the graph contains no (reachable) strictly negative cycle. If~the
  graph also contains no zero cycle, then the energy level necessarily
  diverges to~$+\infty$, and the average-energy is~$+\infty$ along any run. Indeed, we are in the case of Point~\ref{prop:MPpos} of Lem.~\ref{lem:MPInfToAE}. Any strategy is optimal in that case: in particular, any memoryless strategy is.

 For the rest of this proof, we consider the remaining case of graphs that contain no reachable strictly negative
  cycle, but that do contain zero cycles. We will prove that memoryless strategies suffice for $\playerOne$ in those games, by induction on the number of choices of $\playerOne$. Given a game $\Game =
(S_1, S_2 = \emptyset, \trans, \weg)$, we define $d_\Game = \size{\trans} - \size{S}$. Since we assume graphs to be deadlock-free, we have that $d_\Game \geq 0$ for any game $\Game$. We consider induction on the value $d_\Game$. For every game $\Game$ such that $d_\Game = 0$ and initial state $\initState \in S$, $\playerOne$ wins for the $\AE$ objective for threshold $t \in \mathbb{Q}$ iff he wins with a memoryless strategy: indeed, $\playerOne$ actually has no choice at all in $\Game$, which is reduced to a unique outcome from $\initState$.

Now assume that memoryless strategies suffice for $\playerOne$ in every game $\Game$ such that $d_\Game \leq m$ for some $m \in \mathbb{N}$. We claim that they also suffice in every game $\Game$ such that $d_\Game = m+1$. Observe that if this holds, we are done as it proves that memoryless strategies suffice for $\playerOne$ in all one-player $\AE$ games. Let $\Game$ be such a game with $d_\Game = m+1$. Recall that in $\Game =
(S_1, S_2 = \emptyset, \trans, \weg)$ there is no strictly negative cycle by hypothesis. Let $s$ be a state of $\Game$ such that $s$ has at least two outgoing edges. Such a state necessarily exists since $d_\Game \geq 1$. Consider a partition of the outgoing edges of $s$ in two non-empty sets $A$, $B$ such that $A \uplus B = \{(s_1, s_2) \in \trans \mid s_1 = s\}$. According to this partition, we can define in the natural way two sub-games $\Game_A=
(S_1, S_2 = \emptyset, \trans \setminus B, \weg)$ and $\Game_B =
(S_1, S_2 = \emptyset, \trans \setminus A, \weg)$ such that $d_{\Game_A} \leq m$ and $d_{\Game_B} \leq m$. By induction hypothesis, we know that memoryless strategies suffice to play optimally for the $\AE$ objective in those two sub-games. First, observe that if $\playerOne$ has a memoryless winning strategy $\sigma$ in either $\Game_A$ or $\Game_B$ for threshold $t \in \mathbb{Q}$, then this strategy remains winning in $\Game$. What we need to show is that if $\playerOne$ cannot win in both $\Game_A$ and $\Game_B$, then he also cannot win in $\Game$, even using memory in $s$: in the following, we assume that $\playerOne$ is memoryless in any other state $s' \neq s$ (following the induction hypothesis) and we show that mixing cycles in $s$ does not help him.

By contradiction, assume that $\playerOne$ cannot win in both $\Game_A$ and $\Game_B$, but he has a winning strategy $\sigma$ in $\Game$, for the same threshold $t$. Let $\play$ be the outcome consistent with $\sigma$. Two cases are possible.

First, state $s$ is seen \textit{finitely often} along $\play$. In this case, we apply Lemma~\ref{lem:AE_prefix} repeatedly on $\play$ to iteratively remove all cycles on $s$. Since there is no strictly negative cycle in $\Game$, we know that removing one cycle cannot increase the average-energy of the play (it either stays the same if the cycle is a zero cycle, or decreases if it is a strictly positive one). Since $s$ is seen finitely often, we eventually obtain a play $\play'$ that sees $s$ at most once. Therefore, this play either belongs to $\Game_A$ or $\Game_B$ (both if $s$ is never visited). Furthermore, it has average-energy at most $t$ by construction. This contradicts the claim that $\playerOne$ has no winning strategy in both sub-games and concludes the proof in this case.

Second, state $s$ is seen \textit{infinitely often} along $\play$. Since $\playerOne$ is memoryless outside $s$, $\play$ only contains simple cycles and can be written as $\play = \rho \cdot \cycle_{1}
  \cdot \cycle_{2} \cdot \cycle_{3} \cdots{}$ where $\rho$ is an acyclic prefix ending in $s$ and for all $i \geq 1$, $\cycle_{i}$ is a simple cycle on $s$. Observe that every cycle $\cycle_{i}$ belongs either to $\Game_A$ or to $\Game_B$. Furthermore, since $\play$ is winning and there is no strictly negative cycle in $\Game$, only finitely many indices $i_1$, \ldots{}, $i_k$ may correspond to a strictly positive cycle. With the same reasoning as above (repeated application of Lemma~\ref{lem:AE_prefix}), we have that the play $\play' = \rho \cdot \cycle_{i_k+1} \cdot \cycle_{i_k +2}\cdots{}$, obtained by removing the first cycles up to index $i_k$, necessarily has a lower or equal average-energy: hence it is also winning. Now observe that the sequence of cycles $\play'' = \cycle_{i_k+1} \cdot \cycle_{i_k +2}\cdots{}$ may still involve simple cycles from both $\Game_A$ and $\Game_B$. Still, as all cycles are of length at most $\size{S}$, and are zero cycles, we can apply Lemma~\ref{lem:AE_repeat} to extract one best cycle $\cycle_{j}$, $j > i_k$. Putting all this together, we have that $\play''' = \rho \cdot (\cycle_j)^{\omega}$ is such that $\AEsup(\play''') \leq \AEsup(\play)$. Furthermore, $\play'''$ is a simple lasso path that belongs either to $\Game_A$ or to $\Game_B$ (as it now uses a unique outgoing edge from $s$). Consequently, $\play'''$ describes a winning strategy in one of the sub-games, which contradicts our hypothesis and concludes our proof in this case too.
\end{proof}

\paragraph{Polynomial-time algorithm.} We now know the form of optimal memoryless strategies: an optimal lasso path $\play = \prefix \cdot \calC^\omega$ w.r.t.~the $\AE$. We establish a polynomial-time algorithm to solve one-player $\AEG$ games.

The crux of our algorithm  consists in
  computing, for each state~$s$, the best\,---\,w.r.t.~the~$\AE$\,---\,\textit{zero} cycle~$\calC_s$ starting and ending in~$s$ (if~any). This is achieved through linear programming~(LP) over expanded graphs. For~each state~$s$ and length $k \in \{1, \ldots{}, \vert \states \vert\}$, we compute the best cycle $\calC_{s, k}$ by considering a graph (Fig.~\ref{fig:LPGraphForAE}) that models all cycles of length~$k$ from~$s$ and that uses $k+1$ levels and two-dimensional weights on edges of the form $(c, l\cdot c)$ where $c$ is the weight in the original game and $l \in \{k, k-1, \ldots{}, 1\}$ is the level of the edge. In~the~LP, we~look for cycles $\calC_{s,k}$ of length $k$ on $s$ such that (a)~the~sum of weights in the first dimension is zero (thus $\calC_{s,k}$ is a \textit{zero} cycle), and (b)~the~sum in the second one is minimal. Fortunately, this sum is exactly equal to $\AE(\calC) \cdot k$ thanks to the $l$ factors used in the weights of the expanded graph. Hence, we obtain the optimal cycle $\calC_{s, k}$ (in~polynomial time).
Doing this $\vert \states \vert$ times for each state~$s$, we~obtain for each of them the optimal cycle $\calC_{s}$ (if one zero cycle exists). Then, by
  Lem.~\ref{lem:AE_prefix}, it~remains to compute the least~$\EL$
  with which each state~$s$ can be reached using classical graph techniques (e.g., Bellman-Ford), and to pick the optimal combination to obtain an optimal memoryless strategy, in polynomial time.

\begin{figure}[htb]
        \centering
\subfloat[Original game.]{\scalebox{1}{\begin{tikzpicture}[->,>=stealth',shorten >=1pt,auto,node
    distance=2.5cm,bend angle=45, scale=0.5, font=\normalsize,inner sep=.5mm]
    \everymath{\scriptstyle}
    \tikzstyle{p1}=[draw,circle,text centered,minimum size=7mm,text width=5mm]
    \tikzstyle{p2}=[draw,rectangle,text centered,minimum size=7mm,text width=4mm]
    \node[p1]  (0)  at (0, 0) {$s'$};
    \node[p1]  (1) at (3, 0) {$s$};
    \node[p1]  (2) at (6, 0) {$s''$};
    
    \coordinate[shift={(0mm,5mm)}] (init) at (1.north);
    \path
    (init) edge (1);
	\draw[->,>=latex] (0) to[out=40,in=140] node[above] {$1$} (1);
	\draw[->,>=latex] (1) to[out=40,in=140] node[above] {$1$} (2);
	\draw[->,>=latex] (2) to[out=220,in=-40] node[below] {$-1$} (1);
	\draw[->,>=latex] (1) to[out=220,in=-40] node[below] {$-1$} (0);
      \end{tikzpicture}}}
      \hspace{2cm}
      \subfloat[Expanded graph for $k = 2$.]{\scalebox{1}{\begin{tikzpicture}[->,>=stealth',shorten >=1pt,auto,node
    distance=2.5cm,bend angle=45, scale=.8, font=\small,inner sep=.5mm]
    \everymath{\scriptstyle}
    \tikzstyle{p1}=[draw,ellipse,text centered,minimum size=7mm,text width=8.5mm]
    \tikzstyle{p2}=[draw,rectangle,text centered,minimum size=7mm,text width=4mm]
    \node[p1]  (0)  at (0, 0) {$(s, 2)$};
    \node[p1]  (1) at (3, 0.6) {$(s', 1)$};
    \node[p1]  (2) at (3, -0.6) {$(s'', 1)$};
    \node[p1]  (4) at (6, 0) {$(s, 0)$};
    \coordinate[shift={(-5mm,0mm)}] (init) at (0.west);
    \path
    (0) edge node[above,xshift=-2mm] {$(-1, -2)$} (1)
    (0) edge node[below,xshift=-2mm] {$(1, 2)$} (2)
    (1) edge node[above,xshift=2mm] {$(1, 1)$} (4)
    (2) edge node[below,xshift=2mm] {$(-1, -1)$} (4)
    (init) edge (0);
      \end{tikzpicture}}}
	\caption{The best cycle $\calC_{s, 2}$ is computed by looking for a
          path from $(s,2)$ to $(s,0)$ with sum zero in the first dimension
          (zero cycle) and minimal sum in the second dimension (minimal
          $\AEG$). Here, the cycle via $s'$ is clearly better, with $\AEG$
          equal to $-1/2$ in contrast to $1/2$ via $s''$.}
	\label{fig:LPGraphForAE}
\end{figure}
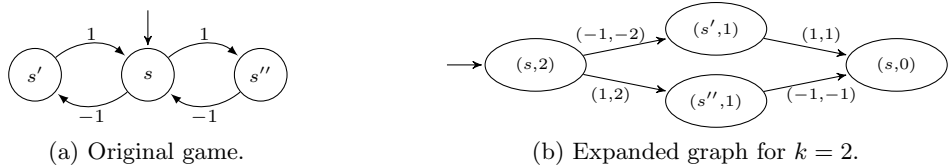

\begin{theorem}
\label{thm:ae_onePlayer_PTIME}
  The $\AEG$ problem for one-player games is in \PTIME.	
\end{theorem}

\begin{proof}
Let $\initState$ be the initial state and $t \in \mathbb{Q}$ be the threshold.
  From Thm.~\ref{thm:aeg_memoryless}, we~can restrict our search to \textit{memoryless} strategies achieving average-energy less than or equal to~$t$. As noted in the proof of Thm.~\ref{thm:aeg_memoryless}, if a strictly negative simple cycle exists and can be reached from $\initState$, then the answer to the $\AEG$ problem is clearly $\textsf{Yes}$, as average-energy $-\infty$ is achievable. Checking if such a cycle exists and is reachable can be done in cubic time in the number of states (e.g., using Bellman-Ford to detect negative cycles).
  
Hence, we now assume that no negative cycle exists. The main part of our algorithm consists in
  computing, for each state~$s$, the least average-energy that can be achieved
  along a simple \textit{zero} cycle starting and ending in~$s$ (if~any). Indeed, strictly positive cycles should be avoided as there is no negative cycle to counteract them. Applying
  Lem.~\ref{lem:AE_prefix}, it~then remains to compute the least energy level
  with which each state $s$ can be reached (simple paths are sufficient as there
  are no negative cycles), and to pick the optimal combination. Again, this last part can be solved by using classical graph algorithms in cubic time in $\vert S\vert$.

  We now focus on computing the best zero cycle from a state~$s$. This is achieved by enumerating
  the possible lengths, from~$1$ to~$\size S$ (\textit{simple} cycles suffice). For a fixed length~$k$, 
  we~consider a new graph~$\calG_{s,k}$, made of $k+1$ copies of the original game~$\Game$. The
  states of~$\calG_{s,k}$ are pairs~$(u,l)$ with $u\in S$ and~$0\leq l\leq k$. The new graph is arranged in levels, indexed from $l = k$ for the top one to $l = 0$ for the bottom one: $l$ represents the number of steps remaining to close the cycle of length $k$.
  For each edge~$(u,u')$ of~$\Game$, with $w(u,u') = c$, and for each~$1\leq l\leq k$,
  except if both $u'=s$ and~$l<k$ (in~order to rule out intermediary visits to~$s$), 
  there is an edge from $(u,l)$ to $(u',l-1)$. This edge carries a
  pair of weights $(c,l\cdot c)$. Our aim is to find a path in this graph
  from~$(s,k)$ to~$(s,0)$ (hence this is a simple cycle of length~$k$) such that the
  sum of the weights on the first dimension is~zero (hence this is a
  zero cycle) and the sum on the second dimension is minimized (when
  divided by~$k$, this sum is precisely the average-energy, if~starting
  from energy level~zero). 

  This problem can be expressed as a linear program, with
  variables~$x_{u,u',l}$ for each edge~$u\to u'$ and each~$1\leq l\leq
  k$. While they are not required to take integer values, these variables are
  intended to represent the number of times the edge from~$(u,l)$
  to~$(u',l-1)$ is taken along a ``path'' in~$\calG_{s,k}$. The linear
  program is as follows:
\begin{center}
\begin{minipage}{.8\linewidth}
minimize $\sum x_{u,u',l}\cdot l\cdot w(u, u')$ subject to
\begin{enumerate}
\item $0\leq x_{u,u',l}\leq 1$ for all~$x_{u,u',l}$;
\item\label{lp2} for all~$(u,l)$ with $1\leq l\leq k-1$, 
  \(
  \sum_{u'} x_{u',u,l+1} = \sum_{u'} x_{u,u',l} 
  \);
\item\label{lp3}
  \(
  \sum_{u'} x_{s,u',k} = \sum_{u'} x_{u',s,1} =1 
  \);
\item\label{lp4} $\sum x_{u,u',l}\cdot w(u,u')=0$;
\item\label{lp5} $\sum x_{u,u',l} \geq 1$.
\end{enumerate}
\end{minipage}
\end{center}
Condition~\eqref{lp2} states that each state has the same amount of
``incoming'' and ``outgoing'' flow. Condition~\eqref{lp3} expresses the fact
that we start and end up in state~$s$. Condition~\eqref{lp4} encodes the fact
that we are looking for zero cycles, and Condition~\eqref{lp5} rules out the
(possible) trivial solution where all variables are~zero.

First observe that if this LP has no solution, then there is no zero cycle of length $k$ from $s$. Now, assume it has a solution~$(x_{u,u',l}^0)$: this solution minimizes $\sum
x_{u,u',l}\cdot l\cdot w(u,u')$. Consider a sequence of edges $s=u_k
\to u_{k-1} \to \cdots \to u_1 \to u_0=s$ for which $x_{u_l,u_{l-1},l}>0$ for
all~$l$. The existence of such a sequence easily follows from
Conditions~\eqref{lp2} and~\eqref{lp3}. Assume that this is not a zero cycle.
As there are no negative cycles, then this must be a positive cycle. But in
order to fulfill Condition~\eqref{lp4}, we~would need a negative cycle to
compensate for this positive cycle, hence implying contradiction. We~conclude
that any sequence of consecutive edges as selected above is a
zero cycle. Similarly, there cannot be a zero cycle of length~$k$ from~$s$
with better average-energy, as this would contradict the optimality of
this solution. We~thus have obtained an average-energy-optimal simple zero
cycle of length~$k$ from $s$, in polynomial time. Indeed, the LP is polynomial in the size of $\calG_{s,k}$, itself polynomial in the size of the original game: the expanded graph has its size bounded by $\vert \states \vert \cdot (k+1)$ and all weights are bounded by $k\cdot W$ with $k \leq \vert \states\vert$ and $W$ the largest absolute weight in the original game.

As discussed above, this process can be repeated for each state $s$ and each length $k$, $1 \leq k \leq \vert \states \vert$, hence at most $\vert \states \vert^{2}$ times. For each state, we select the best cycle among the $\vert S\vert$ possible ones (one for each length). Therefore, in polynomial time, we get a description of the best cycles w.r.t.~the average-energy, for each $s \in \states$. Clearly if no such cycle exists, then the answer to the $\AEG$ problem is \textsf{No}, as all cycles are strictly positive and the average-energy of any play will be $+\infty$. If some exist, we can find an optimal strategy by picking the best combination between such a cycle from a state $s$ and a corresponding prefix from $\initState$ to $s$ of minimal energy level. As presented before, this is achieved in polynomial time. Then the answer to the $\AEG$ problem is \textsf{Yes} if and only if this optimal combination yields average-energy at most equal to $t$. This concludes our proof.
\end{proof}

\subsection{Two-player games}
\label{subsec:ae_two}
\paragraph{Memoryless determinacy.} We now prove that memoryless strategies still suffice in two-player games. As discussed in Sect.~\ref{subsec:ae_tech}, most classical criteria do not apply. There is, however, one result that proves particularly useful. Consider any payoff function such that memoryless strategies suffice for \textit{both} \textit{one-player} versions ($S_{1} = \emptyset$, resp. $S_{2} = \emptyset$). In~\cite[Cor.~7]{GZ05}, Gimbert and Zielonka establish that memoryless strategies also suffice in \textit{two-player} games with the same payoff. Thanks to Thm.~\ref{thm:aeg_memoryless}, this entails the next theorem.

\begin{theorem}
\label{thm:ae_two_memoryless}
Average-energy games are determined and both players have memoryless optimal strategies.
\end{theorem}

Observe that this result is true for both variants of the average-energy payoff function, namely $\AEsup$ and $\AEinf$. When both players play optimally, they can restrict themselves to memoryless strategies and both variants thus coincide as mentioned earlier.

\paragraph{Solving average-energy games.} Finally, consider the complexity of deciding the winner in a two-player $\AEG$ game. By Thm.~\ref{thm:ae_two_memoryless}, one can guess an optimal memoryless strategy for $\playerTwo$ and solve the remaining one-player game for $\playerOne$, in polynomial time (by Thm.~\ref{thm:ae_onePlayer_PTIME}). The converse is also true: one can guess the strategy of $\playerOne$ and solve the remaining game where $S_{1} = \emptyset$ in polynomial time. Thus, we obtain the following result.

\begin{theorem}
\label{thm:ae_npinter}
The $\AEG$ problem for two-player games is in \NP $\cap$ \coNP.
\end{theorem}

We complete our study by proving that $\MPG$ games can be encoded into $\AEG$ ones in polynomial time. The former are known to be in \NP $\cap$ \coNP but whether they belong to $\PTIME$ is a long-standing open question (e.g.,~\cite{ZP96,ipl68(3)-Jur,BCDGR11,Chatterjee201525}). Hence, w.r.t.~current knowledge, the \NP $\cap$ \coNP-membership of the $\AEG$ problem can be considered optimal. The~key of the construction is to double each edge of the original game and modify the weight function such that each pair of successive edges corresponding to such a doubled edge now has a total energy level of zero, and an average-energy that is exactly equal to the weight of the original edge. Then we apply decomposition techniques as in Lem.~\ref{lem:AE_repeat} to establish the equivalence.

\begin{theorem}
\label{thm:mp_to_ae}
Mean-payoff games can be reduced to average-energy games in polynomial time.
\end{theorem}

\begin{proof}
Let~$G=(S_1,S_2,E,w)$ be a game, and $t\in\bbQ$ be the threshold for the mean-payoff problem. From~$G$, we~build
another game~$G'=(S'_1,S'_2,E',w')$ such that
\begin{itemize}
\item $S'_1=S_1\cup E$ and $S'_2=S_2$;
\item $E'$ contains two types of edges: 
  \begin{itemize}
  \item $(s,e)\in E'$ iff there exists $s'$ such that $e=(s,s')\in E$. Then
    $w'(s,e) = 2\cdot w(e)$.
  \item $(e,s') \in E'$ for any~$e=(s,s') \in E$. Then $w'(e,s') = -2\cdot w(e)$.
  \end{itemize}
\end{itemize}
We claim that $\playerOne$ has a strategy ensuring objective $\MeanPayOff(t)$ in $G$ if and only if the answer for the $\AEG$ problem in $G'$ is \textsf{Yes} for the same threshold $t$. A similar construction is used in~\cite{BEGM15}.

With a prefix~$\rho=(s_i)_{i\leq n}$ in~$G$, we~can associate a prefix
$\rho'=(s'_i)_{i\leq 2n}$ in~$G'$ as follows: for all $k \leq n$, $s'_{2k}=s_k$, and for all $k < n$,
$s'_{2k+1}=(s_k,s_{k+1})$. The mean-payoff along~$\rho$ then equals the
average energy along~$\rho'$ (assuming initial energy~$0$ for~$\rho'$). Indeed, applying the same decomposition arguments as for Lem.~\ref{lem:AE_repeat} and by definition of the weight function $w'$, we have that
\begin{align*}
\AE(\rho') &= \dfrac{1}{n}\sum_{i = 0}^{n-1} \dfrac{2\cdot w'(s_i, (s_i, s_{i+1})) + w'((s_i, s_{i+1}), s_{i+1})}{2}\\ &= \dfrac{1}{n}\sum_{i = 0}^{n-1} \dfrac{4\cdot w(s_i, s_{i+1}) -2\cdot w(s_i, s_{i+1})}{2} = \dfrac{1}{n}\sum_{i = 0}^{n-1} w(s_i, s_{i+1}) = \MP(\rho).
\end{align*}
Conversely, with a prefix~$\rho'=(s'_i)_{i\leq 2n}$ in~$G'$ starting and ending in a state
in~$S_1\cup S_2$, we~can associate a prefix $\rho=(s_i)_{i\leq n}$ in $G$ such that
$s_k=s'_{2k}$ for all $k \leq n$. Again, assuming the initial energy is~zero in~$\rho'$, the
average energy along~$\rho'$ equals the mean payoff along~$\rho$.

Now, assume that $\playerOne$ has a winning strategy~$\sigma$ in $G$ from some
state~$s\in S_1\cup S_2$, achieving mean-payoff less than or equal to~$t$.
Consider the strategy~$\sigma'$ for $G'$ defined as $\sigma'(\rho') = \sigma(\rho)$
if~$\rho'$ ends in~$S_1$. If~$\rho'$ ends in a $T$-state of the
form~$(s,s')$, then we let $\sigma'(\rho')=s'$, which is the only possible
outgoing edge. We see that the outcomes of~$\sigma'$
correspond to the outcomes of~$\sigma$, so that, assuming that the initial
energy level is~zero, $\sigma'$ enforces that the average-energy is
below~$t$ for any infinite outcome. Conversely, given a strategy~$\sigma'$ for $G'$
whose outcomes have average-energy below~$t$, the strategy defined by
$\sigma(\rho) = \sigma'(\rho')$ for all finite paths~$\rho$ in~$G$ secures a
mean-payoff below~$t$. Observe that the equivalence holds both between $\AEsup$ and $\MPsup$, and between $\AEinf$ and $\MPinf$. Indeed, we have seen that for both \textit{MP} and $\AEG$ games, memoryless strategies suffice and decision problems for both variants coincide.
\end{proof}

\section{Average-Energy with Lower- and Upper-Bounded Energy}
\label{sec:average_lu}
We extend the $\AEG$ framework with constraints on the running energy level of the system. Such constraints are natural in many applications where the energy capacity is bounded (e.g., fuel tank, battery charge). We first study the case where the energy is subject to \textit{both} a lower bound (here, zero) \textit{and} an upper bound ($U \in \mathbb{N}$). 
We study the problem for the \textit{fixed initial energy level} $\initCredit \coloneqq 0$. In this case, the range of acceptable energy levels is by definition constrained to the interval $[ 0, U]$. Our approach benefits from this: we solve the $\AELU$ problem by considering an $\AEG$ problem (and subsequently, an $\MPG$ problem) over an expanded game that explicitly accounts for the lower and upper bounds on the energy.

Formally, we want to decide if \playerOne can ensure a \textit{sufficiently low} $\AE$ while keeping the $\EL$ within the allowed range.

\begin{bproblem}[$\AELU$] 
Given a game~$\Game$, an initial state $\initState$, an upper bound $U \in \mathbb{N}$, and a threshold~$t \in \mathbb{Q}$, decide if $\pI$ has a winning strategy $\St_1 \in \strats_{1}$ for the objective $\LUBound(U, \initCredit \coloneqq 0)\,\cap\, \AvgEnergyLevel(t)$.
\end{bproblem}

Again, we present results for the supremum variant $\AEsup$ but they also hold for the infimum one $\AEinf$.

\paragraph{Illustration.} Consider the one-player game in Fig.~\ref{fig:aelu_example}. The energy constraints force $\playerOne$ to keep the energy in $[0,\,3]$ at all times. Hence, only three strategies can be followed safely, respectively inducing plays $\play_{1}$, $\play_{2}$ and $\play_{3}$. Due to the bounds on energy, it is natural that strategies need to alternate between both a positive and a negative cycle to satisfy objective $\LUBound(U, \initCredit \coloneqq 0)$ (since no simple zero cycle exists). It is yet interesting that to play optimally (play $\play_{3}$), $\playerOne$ actually has to use \textit{both} positive cycles, and in the \textit{appropriate order} (compare plays $\play_{2}$ and $\play_{3}$).

\colorlet{newgreen}{green!60!black}

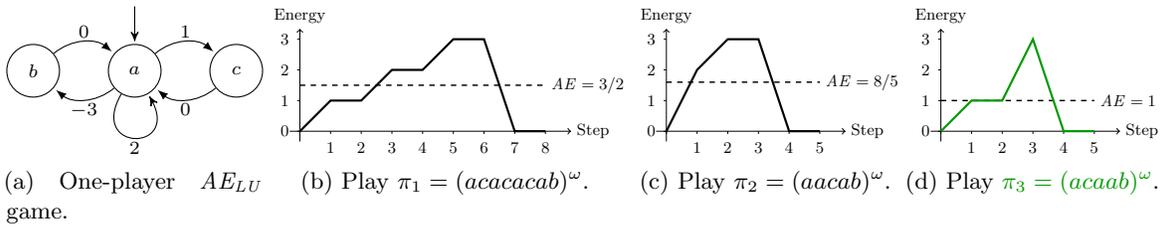
\begin{figure}[htb]
	\centering
	\subfloat[One-player $\AELU$ game.]{\label{fig:aelu_game}\scalebox{1}{\begin{tikzpicture}[->,>=stealth',shorten >=1pt,auto,node
    distance=2.5cm,bend angle=45, scale=0.6, font=\normalsize,scale=.75,inner sep=0.5mm]
    \everymath{\scriptstyle}
    \tikzstyle{p1}=[draw,circle,text centered,minimum size=7mm,text width=4mm]
    \tikzstyle{p2}=[draw,rectangle,text centered,minimum size=7mm,text width=4mm]
    \node[p1]  (0)  at (0, 0) {$b$};
    \node[p1]  (1) at (3, 0) {$a$};
    \node[p1]  (2) at (6, 0) {$c$};
    
    \coordinate[shift={(0mm,5mm)}] (init) at (1.north);
    \path
    (init) edge (1)
    (1) edge [loop below, out=240, in=300,looseness=2, distance=2cm] node [below] {$2$} (1);
	\draw[->,>=latex] (0) to[out=40,in=140] node[above] {$0$} (1);
	\draw[->,>=latex] (1) to[out=40,in=140] node[above] {$1$} (2);
	\draw[->,>=latex] (2) to[out=220,in=-40] node[below] {$0$} (1);
	\draw[->,>=latex] (1) to[out=220,in=-40] node[below] {$-3$} (0);
      \end{tikzpicture}}}	
	\subfloat[Play $\play_{1} = (acacacab)^{\omega}$.]{\scalebox{0.68}{
	\begin{tikzpicture}

		\def \xscale {0.6}
		\def \yscale {0.6}
		\def \xmax {8}
		\def \ymax {3}
			
		\def\xy#1#2{({#1*\xscale},{#2*\yscale})}
			   
	    \draw[->] (-0.2,0) -- ({\xmax*\xscale+0.5},0) node[right] {Step};
	    \draw[->] (0,-0.2) -- (0,{\ymax*\yscale + 0.2}) node[above] {Energy};
 		
        \foreach \y in {0,1,...,\ymax} {	
			\draw (-0.1,{\y*\yscale}) node[anchor=east] {\y} ;
		}
		
        \foreach \x in {1,2,...,\xmax} {
			\draw ({\x*\xscale},-0.1) node[anchor=north] {\x} ;
		}

        \foreach \x in {0,...,\xmax} {
            \draw ({\x*\xscale},0) -- ({\x*\xscale},-1.5pt);
        }

        \foreach \y in {0,...,\ymax} {
            \draw (0, {\y*\yscale}) -- (-1.5pt, {\y*\yscale});
        }
		
		\def\mean{3/2}
		\draw[thick,dashed,color=black]
			\xy{0}{\mean} -- \xy{\xmax}{\mean};
		
		\draw \xy{\xmax}{\mean} node[anchor=west] {$ \mathit{AE} = 3/2$ };

		\draw[very thick,-,color=black]
					\xy{0}{0} --
					\xy{1}{1} --
					\xy{2}{1} --
					\xy{3}{2} --
					\xy{4}{2} --
					\xy{5}{3} --
					\xy{6}{3} --
					\xy{7}{0} --
					\xy{8}{0} ;
					 
	\end{tikzpicture}}}
	\subfloat[Play $\play_{2} = (aacab)^{\omega}$.]{\scalebox{0.68}{
	\begin{tikzpicture}

		\def \xscale {0.6}
		\def \yscale {0.6}
		\def \xmax {5}
		\def \ymax {3}
			
		\def\xy#1#2{({#1*\xscale},{#2*\yscale})}
			   
	    \draw[->] (-0.2,0) -- ({\xmax*\xscale+0.5},0) node[right] {Step};
	    \draw[->] (0,-0.2) -- (0,{\ymax*\yscale + 0.2}) node[above] {Energy};
 		
        \foreach \y in {0,1,...,\ymax} {	
			\draw (-0.1,{\y*\yscale}) node[anchor=east] {\y} ;
		}
		
        \foreach \x in {1,2,...,\xmax} {
			\draw ({\x*\xscale},-0.1) node[anchor=north] {\x} ;
		}

        \foreach \x in {0,...,\xmax} {
            \draw ({\x*\xscale},0) -- ({\x*\xscale},-1.5pt);
        }

        \foreach \y in {0,...,\ymax} {
            \draw (0, {\y*\yscale}) -- (-1.5pt, {\y*\yscale});
        }
		
		\def\mean{8/5}
		\draw[thick,dashed,color=black]
			\xy{0}{\mean} -- \xy{\xmax}{\mean};
		
		\draw \xy{\xmax}{\mean} node[anchor=west] {$ \mathit{AE} = 8/5$ };

		\draw[very thick,-,color=black]
					\xy{0}{0} --
					\xy{1}{2} --
					\xy{2}{3} --
					\xy{3}{3} --
					\xy{4}{0} --
					\xy{5}{0} ;
					 
	\end{tikzpicture}}}
	\subfloat[Play {\color{newgreen}$\play_{3} = (acaab)^{\omega}$}.]{\scalebox{0.68}{
	\begin{tikzpicture}

		\def \xscale {0.6}
		\def \yscale {0.6}
		\def \xmax {5}
		\def \ymax {3}
			
		\def\xy#1#2{({#1*\xscale},{#2*\yscale})}
			   
	    \draw[->] (-0.2,0) -- ({\xmax*\xscale+0.5},0) node[right] {Step};
	    \draw[->] (0,-0.2) -- (0,{\ymax*\yscale + 0.2}) node[above] {Energy};
 		
        \foreach \y in {0,1,...,\ymax} {	
			\draw (-0.1,{\y*\yscale}) node[anchor=east] {\y} ;
		}
		
        \foreach \x in {1,2,...,\xmax} {
			\draw ({\x*\xscale},-0.1) node[anchor=north] {\x} ;
		}

        \foreach \x in {0,...,\xmax} {
            \draw ({\x*\xscale},0) -- ({\x*\xscale},-1.5pt);
        }

        \foreach \y in {0,...,\ymax} {
            \draw (0, {\y*\yscale}) -- (-1.5pt, {\y*\yscale});
        }
		
		\def\mean{1}
		\draw[thick,dashed,color=black]
			\xy{0}{\mean} -- \xy{\xmax}{\mean};
		
		\draw \xy{\xmax}{\mean} node[anchor=west] {$ \mathit{AE} = 1$ };

		\draw[very thick,-,color=newgreen]
					\xy{0}{0} --
					\xy{1}{1} --
					\xy{2}{1} --
					\xy{3}{3} --
					\xy{4}{0} --
					\xy{5}{0} ;
					 
	\end{tikzpicture}}}
	\caption{Example of a one-player $\AELU$ game ($U = 3$) and the evolution of energy under different strategies that maintain it within $\left[0,\, 3\right]$ at all times. The minimal average-energy is obtained with {\color{newgreen}play $\play_{3}$}: alternating in order between the $+1$, $+2$ and $-3$ cycles.}\label{fig:aelu_example}	
\end{figure}

This type of alternating behavior is more intricate than for other classical conjunctions of objectives.
Consider for example energy parity~\cite{CD10} or multi-dimensional energy games~\cite{CRR14,VCDHRR15}. It is usually necessary to use different cycles in such games: intuitively, one needs one ``good'' cycle for each dimension and one for the parity objective, and a winning strategy needs to alternate between those cycles. However, there is no need to use \textit{two} different cycles that are ``good'' w.r.t.~the same part of the objective. In the case of $\AELU$ games, we see that it is sometimes necessary to use two (or more) different cycles even though they impact the sum of weights in the same direction (e.g., several positive cycles). This gives a hint of the complexity of $\AELU$ games.

\subsection{Pseudo-polynomial algorithm and complexity bounds}
\label{subsec:AELU_algo}

We first reduce the $\AELU$ problem to the $\AE$ problem over
a \textit{pseudo-polynomial expanded game}, i.e., polynomial in the size of
the original $\AELU$ game and in $U \in \mathbb{N}$. By
Thm.~\ref{thm:ae_npinter} and Thm.~\ref{thm:ae_onePlayer_PTIME}, this
reduction induces $\NEXPTIME~\cap~\coNEXPTIME$-membership of the two-player
$\AELU$ problem, and $\EXPTIME$-membership of the one-player one. We~improve
the complexity for two-player games by further reducing the $\AE$
game to an $\MP$ game: this yields $\EXPTIME$-membership, which is optimal
(Thm.~\ref{thm:aelu_reduc}). We~also improve the one-player case by observing
that a witness lasso path in the $\MP$ game can be built on-the-fly, and the mean-payoff of this path can be computed using only polynomial space in the original game, hence we~end up with \PSPACE-membership which
we also prove optimal in Thm.~\ref{thm:aelu_reduc}.

Observe that if $U$ is encoded in unary or if~$U$ is polynomial in the size of the original game, the complexity of the $\AELU$ problem collapses to $\NP \cap \coNP$ for two-player games and to $\PTIME$ for one-player games thanks to our reduction to an $\AE$ problem and the results of Thm.~\ref{thm:ae_npinter} and Thm.~\ref{thm:ae_onePlayer_PTIME}.

\paragraph{The reductions.} Given a game $G = (S_{1}, S_{2}, E, w)$, an initial state $\initState$, an upper bound $U \in \mathbb{N}$, and a threshold $t \in \mathbb{Q}$, we reduce the $\AELU$ problem to an $\AEG$ problem as follows. If at any point along a play, the energy drops below zero or exceeds $U$, the play will be losing for the $\LUBound(U, \initCredit \coloneqq 0)$ objective, hence also for its conjunction with the $\AEG$ one. So we build a new game $G'$ over the state space $(S \times \{0, 1, \ldots{}, U\}) \cup \{\textsf{sink}\}$. The idea is to include the energy level within the state labels, with $\textsf{sink}$ as an absorbing state reached only when the energy constraint is breached. We now consider the $\AEG$ problem for threshold $t$ on $G'$. By putting a self-loop of weight $1$ on $\textsf{sink}$, we ensure that if the energy constraint is not guaranteed in $G$, the answer to the $\AEG$ problem in $G'$ will be \textsf{No} as the average-energy will be infinite due to reaching this positive loop and repeating it forever.
Hence, we show that the $\AELU$ objective can be won in $G$ if and only if the $\AE$ one can be won in $G'$ (thus avoiding the $\textsf{sink}$ state). The result of the reduction for the game in Fig.~\ref{fig:aelu_game} is presented in Fig.~\ref{fig:aelu_to_ae}.

\begin{figure}[htb]
\hspace{2.6cm}\scalebox{0.7}{\begin{tikzpicture}[->,>=stealth',shorten >=1pt,auto,node
    distance=2.5cm,bend angle=45, scale=1.2, font=\small]
    \tikzstyle{p1}=[draw,ellipse,text centered,minimum size=9mm,text width=8mm]
    \tikzstyle{p2}=[draw,rectangle,text centered,minimum size=7mm,text width=4mm]
    \node[p1]  (a0)  at (0, 0) {$(a, 0)$};
    \node[p1]  (a1) at (3, 0) {$(a, 1)$};
    \node[p1]  (a2) at (6, 0) {$(a, 2)$};
    \node[p1]  (a3) at (9, 0) {$(a, 3)$};
    \node[p1]  (b0)  at (-0.6, -1.2) {$(b, 0)$};
    \node[p1]  (b1) at (2.4, -1.2) {$(b, 1)$};
    \node[p1]  (b2) at (5.4, -1.2) {$(b, 2)$};
    \node[p1]  (b3) at (8.4, -1.2) {$(b, 3)$};
    \node[p1]  (c0)  at (-2, -2) {$(c, 0)$};
    \node[p1]  (c1) at (1, -2) {$(c, 1)$};
    \node[p1]  (c2) at (4, -2) {$(c, 2)$};
    \node[p1]  (c3) at (7, -2) {$(c, 3)$};
    \node[p1]  (sink) at (4.5, 1.9) {$\textsf{sink}$};
    \coordinate[shift={(-5mm,0mm)}] (init) at (a0.west);
    \path
    (a0) edge[thick,color=newgreen] node[right] {$1 \mid 0$} (c1)
    (a1) edge node[right] {$1 \mid 1$} (c2)
    (a2) edge node[right] {$1 \mid 2$} (c3)
    (b0) edge[thick,color=newgreen] node[right,yshift=-1mm,xshift=-1mm] {$0 \mid 0$} (a0)
    (b1) edge node[right,yshift=-1mm,xshift=-1mm] {$0 \mid 1$} (a1)
    (b2) edge node[right,yshift=-1mm,xshift=-1mm] {$0 \mid 2$} (a2)
    (b3) edge node[right,yshift=-1mm,xshift=-1mm] {$0 \mid 3$} (a3)
    (a0) edge node[left,xshift=-4mm] {$1 \mid 0$} (sink)
    (a1) edge node[left,xshift=-1mm] {$1 \mid 1$} (sink)
    (a2) edge node[right,xshift=1mm] {$1 \mid 2$} (sink)
    (a3) edge node[right,xshift=4mm] {$1 \mid 3$} (sink)
    (sink) edge [loop below, out=245, in=295,looseness=1, distance=0.8cm] node [below] {$1 \mid 2$} (sink)
    (init) edge (a0);
	\draw[->,>=latex] (c0) to[out=60,in=210] node[left,xshift=3mm,yshift=3mm] {$0 \mid 0$} (a0);
	\draw[->,>=latex,thick,color=newgreen] (c1) to[out=60,in=210] node[left,xshift=3mm,yshift=3mm] {$0 \mid 1$} (a1);
	\draw[->,>=latex] (c2) to[out=60,in=210] node[left,xshift=3mm,yshift=3mm] {$0 \mid 2$} (a2);
	\draw[->,>=latex] (c3) to[out=60,in=210] node[left,xshift=3mm,yshift=3mm] {$0 \mid 3$} (a3);
	\draw[->,>=latex,thick,color=newgreen] (a3) -- (10,0) -- (10,-2.5) to[out=180,in=0] node[above,xshift=8mm] {$-3 \mid 3$} (-0.6,-2.5) to[out=90, in=270] (b0);
	\draw[->,>=latex] (a0) to[out=15,in=165] node[above,xshift=-7mm,yshift=-1mm] {$2 \mid 0$} (a2);
	\draw[->,>=latex,thick,color=newgreen] (a1) to[out=15,in=165] node[above,xshift=7mm,yshift=-1mm] {$2 \mid 1$} (a3);
      \end{tikzpicture}}
	\caption{Reduction from the $\AELU$ game in Fig.~\ref{fig:aelu_game} to an $\AE$ game and further reduction to an $\MP$ game over the same expanded graph. For the sake of succinctness, the weights are written as $c \mid c'$ with $c$ the weight used in the $\AE$ game and $c'$ the one used in the $\MP$ game. We use the upper bound $U = 3$ and the average-energy threshold $t = 1$ (the optimal value in this case). The optimal play {\color{newgreen}$\play_{3} = (acaab)^{\omega}$} of the original game corresponds to an optimal memoryless play in the expanded graph.}\label{fig:aelu_to_ae}	
\vspace{-4mm}
\end{figure}

\begin{lemma}
\label{lem:aelu_to_ae}
The $\AELU$ problem over a game $G = (S_{1}, S_{2}, E, w)$, with an initial state $\initState$, an upper bound $U \in \mathbb{N}$, and a threshold $t \in \mathbb{Q}$, is reducible to an $\AE$ problem for the same threshold $t \in \mathbb{Q}$ over a game $G' = (S'_{1}, S'_{2}, E', w')$ such that $\vert \states' \vert = (U + 1) \cdot \vert \states \vert + 1$ and $W' = \max\, \{\min\,\{W,\, U\},\: 1\}$, i.e., the largest absolute weight in $G'$ is at most the same as in $G$, or equal to constant $1$.
\end{lemma}

\begin{proof}
Consider the game $G= (S_{1}, S_{2}, E, w)$, with initial state $\initState$, upper bound $U \in \mathbb{N}$ and threshold $t \in \mathbb{Q}$.  We define the expanded game $G' = (S'_{1}, S'_{2}, E', w')$ as follows.
\begin{itemize}
\item $S'_{1} = (S_{1} \times \{0, 1, \ldots{}, U\}) \cup \{\textsf{sink}\}$.
\item $S'_{2} = S_{2} \times \{0, 1, \ldots{}, U\}$.
\item For all $(u, v) \in E$, $(u, c) \in S'$, we have that:
\begin{enumerate}
\item if $d = c + w(u, v) \in [0, U]$, then $e = \big( (u,c), (v, d)\big) \in E'$ and $w'(e) = w(u, v)$,
\item else $e = \big( (u,c),\, \textsf{sink}\big) \in E'$ and $w'(e) = 1$.
\end{enumerate}
\item $(\textsf{sink}, \textsf{sink}) \in E'$ and $w(\textsf{sink}, \textsf{sink}) = 1$.
\end{itemize}
The game $G'$ starts in state $(\initState, 0)$ and edges are built naturally to reflect the changes in the energy level. Whenever the energy drops below zero or exceeds $U$, we redirect the edge to $\textsf{sink}$, where a self-loop of weight $1$ is repeated forever.

We claim that $\playerOne$ has a winning strategy $\sigma_{1}$ for the $\AELU$ objective in $G$ if and only if he has a winning strategy $\sigma'_{1}$ for the $\AE$ objective in $G'$, for the very same average-energy threshold $t$.

First, consider the left-to-right implication. Assume $\sigma_{1}$ is winning for objective $\LUBound(U, \initCredit~\coloneqq~0) \cap \AvgEnergyLevel(t)$ in $G$. The very same strategy can be followed in $G'$, ignoring the additional information on the energy in the state labels. Precisely, for any prefix $\rho' = (s_0, c_0) (s_{1}, c_1) \ldots{} (s_n, c_n)$ in $G'$, we define $\sigma'_{1}(\rho') = (s, c)$ where $s = \sigma_{1}(\rho)$ for $\rho = s_{0} s_{1} \ldots{} s_n$ and $c = c_n + w(s_n, s)$. Obviously, playing this strategy ensures that the special state \textsf{sink} is never reached, as otherwise it would not be winning for $\LUBound(U, \initCredit \coloneqq 0)$ in $G$, by construction of $G'$. Since all weights are identical in both games except on edges entering the \textsf{sink} state, we have that any consistent outcome $\play'$ of $\sigma'_{1}$ in $G'$ corresponds to a consistent outcome $\play$ of $\sigma_{1}$ in $G$ such that $\AEsup(\play') = \AEsup(\play)$, and conversely. Therefore, $\sigma'$ is clearly winning for objective $\AvgEnergyLevel(t)$ in $G'$.

Second, consider the right-to-left implication. Assume $\sigma'_{1}$ is winning for $\AvgEnergyLevel(t)$ in $G'$. Then this strategy ensures that \textsf{sink} is avoided forever. Otherwise, there would exist a consistent outcome $\play'$ reaching \textsf{sink}, and such that $\AEsup(\play') = \infty > t$ because of the strictly positive self-loop. Thus the strategy would not be winning. Hence by construction of $G'$, this strategy trivially ensures $\LUBound(U, \initCredit \coloneqq 0)$ in $G'$. From $\sigma'_{1}$, we build a strategy~$\sigma_{1}$ in~$G$ in the natural way, potentially integrating the information on the energy within the memory of $\sigma_{1}$. Again, there is a bijection between plays avoiding \textsf{sink} in $G'$ and plays in $G$, such that $\sigma_{1}$ is winning for $\LUBound(U, \initCredit \coloneqq 0) \cap \AvgEnergyLevel(t)$ in $G$.

Hence we have shown the claimed reduction. For the sake of completeness, observe that the reduction holds both for $\AEsup$ and $\AEinf$ variants of the average-energy. It remains to discuss the size of the expanded game. Observe that $\vert S' \vert = (U + 1) \cdot \vert S \vert + 1$. Furthermore, if $W$ is the largest absolute weight in $G$, then $W' = \max\, \{\min\,\{W, U\},\: 1\}$ is the largest one in $G'$. Indeed, $W'$ is upper-bounded by $U$ by construction (as all edges of absolute weight larger than $U$ can be redirected directly to \textsf{sink}) and it is lower-bounded by $1$ due to edges leading to \textsf{sink}. So the state space of $G'$ is polynomial in the state space of $G$ and in the \textit{value} of the upper bound~$U$, while its weights are bounded by either the largest weight $W$, the upper bound~$U$ or constant $1$.
\end{proof}

We now show that the $\AE$ game $G'$ can be further reduced to an $\MP$ game $G''$ by modifying the weight structure of the graph. Essentially, all edges leaving a state $(s, c)$ of $G'$ are given weight $c$ in $G''$, i.e., the current energy level, and the self-loop on \textsf{sink} is given weight $(\lceil t\rceil + 1)$. This modification is depicted in Fig.~\ref{fig:aelu_to_ae}. We claim that the $\AEG$ problem for threshold $t \in \mathbb{Q}$ in $G'$ is equivalent to the $\MPG$ problem for the same threshold in $G''$. Indeed, we show that with our change of weight function, reaching \textsf{sink} implies losing, both in $G'$ for $\AEG$ and in $G''$ for $\MPG$, and all plays that \textit{do not} reach \textsf{sink} have the same value for their average-energy in $G'$ as for their mean-payoff in $G''$.
\begin{lemma}
\label{lem:ae_to_mp}
The $\AEG$ problem over the game $G' = (S'_{1}, S'_{2}, E', w')$ defined in Lem.~\ref{lem:aelu_to_ae} is reducible to an $\MPG$ problem for the same threshold $t \in \mathbb{Q}$ over a game $G'' = (S'_{1}, S'_{2}, E', w'')$ sharing the same state space but with largest absolute weight $W'' = \max \{U,\, \lceil t\rceil + 1\}$, where $U$ is the energy upper bound of the original $\AELU$ problem.
\end{lemma}

\begin{proof}
Let $G' = (S'_{1}, S'_{2}, E', w')$ be the game defined in Lem.~\ref{lem:aelu_to_ae}, as a reduction from the original game $G$ for the $\AELU$ problem with upper bound $U \in \mathbb{N}$ and average-energy threshold $t \in \mathbb{Q}$. We now build the game $G'' = (S'_{1}, S'_{2}, E', w'')$  by simply modifying the weight function of $G'$. The changes are as follows:
\begin{itemize}
\item For all edge $e = ((s,c),\,(s',c')) \in E'$, its weight in $G'$ is $w'(e) = c' - c$ and we now set it to $w''(e) = c$ in $G''$. Recall that by construction of $G'$, the value $c$ represents the current energy level for any prefix ending in $(s,c)$. This is the value we now use for the outgoing edge. Also, this value is constrained in $[0,\,U]$ by definition of $G'$.
\item For all edge $e = ((s,c),\,\textsf{sink}) \in E'$, its weight in $G'$ is $w'(e) = 1$ and we now set it to $w''(e) = c$ in $G''$ for the sake of consistency (the actual value over this type of edges will not matter eventually). 
\item For the self-loop $e = (\textsf{sink},\, \textsf{sink}) \in E'$, its weight in $G'$ is $w'(e) = 1$ and we now set it to $w''(e) = \lceil t\rceil + 1$ in $G''$. That is, reaching \textsf{sink} will imply a mean-payoff higher than the threshold.
\end{itemize}

Before proving the claim, we show that for all plays $\play \in \plays(G') = \plays(G'')$ that \textit{do not} reach \textsf{sink}, we have that $\AEsup_{G'}(\play) = \MPsup_{G''}(\play)$, where the subscript naturally refers to the change of weight function. Let $\play = s'_0s'_1s'_2\ldots{} = (s_0, c_0) (s_1, c_1) (s_2, c_2)\ldots{}$ be such a play, where for all $i \geq 0$, $s'_{i} \in \states'$ and $(s_{i}, c_{i}) \in \states \times [0, U] \cap \mathbb{N}$ is its corresponding label. By definition of $G''$, we have that,
\begin{equation*}
\forall\,  n \geq 0,\quad w''(s'_n, s'_{n+1}) = c_n = \EL_{G'}(\play(n)).
\end{equation*}
Hence by definition of the mean-payoff and the average-energy,
\begin{equation}
\label{eq:reduc}
\MPsup_{G''}(\play) = 
	\limsup_{n \to \infty} \frac{1}{n}  \sum_{i = 0}^{n-1} w''(s'_i,s'_{i+1}) = \limsup_{n \to \infty} \frac{1}{n} \sum_{i = 0}^{n-1} \EL_{G'}(\play(i)) = \AEsup_{G'}(\play).
\end{equation}
For the sake of completeness, observe that this equality does not hold for plays reaching \textsf{sink}, as they have infinite average-energy in $G'$ but finite mean-payoff in $G''$.

We proceed by proving the claim that $\playerOne$ has a winning strategy $\sigma'_{1}$ for the $\AE$ objective in $G'$ if and only if he has a winning strategy $\sigma''_{1}$ for the $\MP$ objective in $G''$, for the very same threshold $t$.

First, consider the left-to-right implication. Assume $\sigma'_{1}$ is winning for $\AvgEnergyLevel(t)$ in $G'$. We apply the same strategy in $G''$ straightforwardly as the underlying graph is not modified. Since this strategy is winning for the $\AE$ objective in $G'$, it necessarily avoids \textsf{sink} both in $G'$ and $G''$ (as otherwise the $\AE$ would be infinite). Hence by Eq.~\eqref{eq:reduc}, we have that $\sigma'_{1}$ is also winning for $\MeanPayOff(t)$ in $G''$.

Second, consider the right-to-left implication. Assume $\sigma''_{1}$ is winning for $\MeanPayOff(t)$ in $G''$. Since the self-loop on \textsf{sink} has weight $\lceil t\rceil + 1$, it is necessary that $\sigma''_{1}$ never reaches \textsf{sink} otherwise it would not be winning. Hence we apply the same strategy in $G'$ and by Eq.~\eqref{eq:reduc}, we have that $\sigma''_{1}$ is also winning for $\AvgEnergyLevel(t)$ in $G'$.

This proves correctness of the reduction. The same reasoning can be followed for $\AEinf$ (thus using $\MPinf$) instead of $\AEsup$. We end by discussing the size of $G''$. Clearly, the state space $\states''$ is identical to $\states'$, hence $\vert S'' \vert = (U + 1) \cdot \vert S \vert + 1$. However, the largest absolute weight in $G''$ is $W'' = \max \{U,\, \lceil t\rceil + 1\}$. Indeed, the self-loop on \textsf{sink} has weight $(\lceil t\rceil + 1)$ and all other edges have weight bounded by the energy upper bound $U$ by construction.
\end{proof}

\paragraph{Illustration.} Consider the $\AELU$ game $G$ depicted in Fig.~\ref{fig:aelu_game}. We have seen that the optimal strategy is $\play_{3} = (acaab)^{\omega}$. Now consider the reduction to the $\AE$ game, and further to the $\MP$ game, depicted in Fig.~\ref{fig:aelu_to_ae}. The optimal (memoryless) strategy in both the $\AE$ game $G'$ and the $\MPG$ game $G''$ is to create the play $\play' = ((a,0) (c,1) (a,1) (a,3) (b,0))^{\omega}$, which corresponds to the optimal play $\play_{3}$ in the original game. It can be checked that $\AEsup_{G}(\play_{3}) = \AEsup_{G'}(\play') = \MPsup_{G''}(\play')$. 

\paragraph{Complexity.} 
The reduction from the $\AELU$ game to the $\AE$ one induces a
pseudo-polynomial blow-up in the number of states. Thanks to the second
reduction and the use of a pseudo-polynomial algorithm for the $\MP$
game~\cite{ZP96,BCDGR11}, we get \EXPTIME-membership, which is optimal for
two-player games thanks to the lower bound proved for $\EGLU$~\cite{BFLMS08}.
The complexity is reduced when the bound~$U$ is given in unary or is
polynomial in the size of the game, matching the one obtained for $\AE$ games
without energy constraints.

For the one-player case, we also use the reduction to an $\MP$ game. By~\cite{EM79}, optimal memoryless strategies exist, hence it~suffices to non-deterministically build a simple lasso path in $G''$, and to check that it satisfies the
mean-payoff constraint. It can be done using only polynomial space through on-the-fly computation.

\begin{theorem}
\label{thm:aelu_reduc}
The $\AELU$ problem is $\EXPTIME$-complete for two-player games and 
\PSPACE-complete for one-player games. If the upper bound $U \in \mathbb{N}$ is
polynomial in the size of the game or encoded in unary, the $\AELU$ problem
collapses to $\NP \cap \coNP$ and $\PTIME$ for two-player and one-player games,
respectively.
\end{theorem}

\begin{proof}
Let $G = (S_1, S_2, E, w)$ be the original $\AELU$ game, $W \in \mathbb{N}$ its largest absolute weight, $U \in \mathbb{N}$ the upper bound for energy and $t \in \mathbb{Q}$ the threshold for the $\AELU$ problem. By Lem.~\ref{lem:aelu_to_ae}, this $\AELU$ problem is reducible to an $\AE$ problem for the same threshold $t$ over a game $G' = (S'_{1}, S'_{2}, E', w')$ such that $\vert \states' \vert = (U + 1) \cdot \vert \states \vert + 1$ and $W' = \max\, \{\min\,\{W,\, U\},\: 1\}$. By Lem.~\ref{lem:ae_to_mp}, the $\AELU$ problem can be further reduced to an $\MPG$ problem for the same threshold $t$ over a game $G'' = (S'_{1}, S'_{2}, E', w'')$ sharing the same state space as~$G'$ but with largest absolute weight $W'' = \max \{U,\, \lceil t\rceil + 1\}$. We start by proving the complexity upper bounds.

First, consider the one-player case. Combining Thm.~\ref{thm:ae_onePlayer_PTIME} and the reduction to an $\AE$ game, we obtain that one-player $\AELU$ games can be solved in pseudo-polynomial time, i.e., polynomial in $\vert S \vert$ but also in the value of $U$ (hence exponential in the size of its binary encoding). This both gives $\EXPTIME$-membership of one-player $\AELU$ games with arbitrary upper bounds, and $\PTIME$-membership of the same games with polynomial or unary upper bounds. For arbitrary bounds, we improve the complexity from \EXPTIME to \PSPACE. To do so, we consider the further reduction to an $\MP$ game, but we \textit{do not} completely build the $\MP$ game $G''$ which is known to be of exponential size. Instead, we build non-deterministically a witness lasso path (thanks to memoryless determinacy~\cite{EM79}, they are sufficient) and check on-the-fly that the path is winning or not, using only polynomial space. Recall that we consider a game $G''$ such that $S'_{2} = \emptyset$. Our non-deterministic algorithm answers \textsc{Yes} if $\playerOne$ has a winning strategy in $G''$ (and hence in $G$ thanks to Lem.~\ref{lem:aelu_to_ae} and Lem.~\ref{lem:ae_to_mp}), \textsc{No} otherwise, and is as follows:
\begin{enumerate}
\item Guess a state $s'_{r} \in S'_{1} = (S_{1} \times \{0, 1, \ldots{}, U\}) \cup \{\textsf{sink}\}$ that will be the starting (and ending) state of the cycling part of the lasso path. For the following, we assume that $s'_{r} \neq \textsf{sink}$ otherwise the lasso path that we are trying to build is clearly losing (see proof of Lem.~\ref{lem:ae_to_mp}) and the algorithm answers \textsc{No}. Thus, store state $s'_{r} = (s_{r}, m)$ for some $m \in \{0, \ldots{}, U\}$.
\item Check that $s'_{r}$ is reachable from the initial state $(\initState, 0)$. This can be done in \textsf{NLOGSPACE} w.r.t.~the size of $G''$ (see e.g.,~\cite{DBLP:books/daglib/0086373}), hence \textsf{NPSPACE} w.r.t.~the original problem. If it is not, then the answer is \textsc{No}.
\item Build step by step\footnote{Observe that given a state in $G''$, it is indeed possible to build any neighboring state using only $E$ and $w$ from the original game: one can effectively build the graph $G''$ on-the-fly.} a lasso path by constructing a simple cycle in $G''$ starting in $s'_{r}$. This construction is non-deterministic: if at any point, the \textsf{sink} state is reached, the algorithm returns \textsc{No}. The construction stops as soon as $s'_{r}$ is reached, or after $\vert S'\vert+1$ steps if $s'_{r}$ is not reached: in the latter case, the answer is also \textsc{No} (after $\vert S'\vert +1$ steps, we know for certain that a cycle was created hence our lasso path is complete). While constructing the cycle, we make on-the-fly computations: at each step, the next state is chosen non-deterministically and the only information that is stored~---~except from state $s'_{r}$ used to determine the end of the cycle~---~is the number of steps from leaving $s'_{r}$, and the sum of the weights seen along the cycle.
\item Assume $s'_{r}$ is reached (otherwise we have seen that the answer is \textsc{No}). Let $s'_{0}s'_{1}\ldots{}s'_{l}$ be the sequence of states visited along the construction, with $s'_{0} = s'_{l} = s'_{r}$. We have stored the length $l$ and the sum of weights $\gamma = \sum_{i = 0}^{l-1} w''(s'_{i}, s'_{i+1})$. Now, we check if $\dfrac{\gamma}{l} \leq t$: this quantity is the mean-payoff of the lasso path we have constructed. If yes, then the answer is \textsc{Yes}, thanks to Lem.~\ref{lem:aelu_to_ae} and Lem.~\ref{lem:ae_to_mp}: the lasso path describes a winning strategy. Otherwise, the answer is \textsc{No} as this lasso path represents a losing strategy, by the same lemmas.
\end{enumerate}
The correctness of this algorithm is guaranteed by Lem.~\ref{lem:aelu_to_ae} and Lem.~\ref{lem:ae_to_mp}. It remains to argue that it only uses polynomial space in the original $\AELU$ problem. Observe that our on-the-fly computations only need to record the state $s'_{r}$, the current state, the current length and the current sum. We have that both states belong to $S_{1} \times \{0, 1, \ldots{}, U\}$, that $l < \vert S'\vert = (U + 1) \cdot \vert S \vert + 1$ and that the sum is bounded by $l \cdot W'' = l \cdot \max \{U,\, \lceil t\rceil + 1\}$. Hence, encoding those values only requires a polynomial number of bits w.r.t.~the input of the $\AELU$ problem (i.e., logarithmic in the upper bound $U$, the largest weight $W$ and the threshold $t$). This proves that our algorithm lies in \textsf{NPSPACE}, and by Savitch's theorem~\cite{DBLP:books/daglib/0086373} we know that \textsf{NPSPACE} $= \PSPACE$: hence we proved the upper bound for the one-player $\AELU$ problem.

Second, consider two-player $\AELU$ games. In this case, we solve the $\MPG$ problem over $G''$ using a pseudo-polynomial algorithm such as the one presented in~\cite{BCDGR11}, whose complexity is $\mathcal{O}(\vert S^{\ast}\vert^{3} \cdot W^{\ast})$ for a game with $\vert S^{\ast}\vert$ states and largest absolute weight $W^{\ast} \in \mathbb{N}$. Therefore, the complexity of solving the original $\AELU$ problem is
\begin{equation*}
\mathcal{O}\left(\vert S'\vert^{3} \cdot W''\right) =  \mathcal{O}\left( \big( (U + 1) \cdot \vert \states \vert + 1\big) ^{3} \cdot \max \{U,\, \lceil t\rceil + 1\}\right),
\end{equation*}
which is clearly pseudo-polynomial. Hence we obtain $\EXPTIME$-membership for two-player $\AELU$ games. If the upper bound $U \in \mathbb{N}$ is
polynomial in the size of the game or encoded in unary, it is sufficient to solve the polynomially-larger $\AEG$ game $G'$ using Thm.~\ref{thm:ae_npinter} to obtain $\NP \cap \coNP$-membership.

Now consider lower bounds. The $\AELU$ problem trivially encompasses the lower- and upper-bounded energy problem $\EGLU$, i.e., the $\AELU$ without consideration of the average-energy. Indeed, consider a game~$G$ with an objective $\LUBound(U, \initCredit \coloneqq 0)$, for some $U \in \mathbb{N}$. Assume $\playerOne$ has a winning strategy for this objective. Then this strategy ensures that along any consistent outcome $\play$, the running energy at any point is at most equal to~$U$. By definition, this implies that $\AEinf(\play) \leq \AEsup(\play) \leq U$. Hence this strategy is also winning for the $\AELU$ objective written as the conjunction $\LUBound(U, \initCredit \coloneqq 0) \cap \AvgEnergyLevel(t \coloneqq U)$. The converse is also trivially true.
Ergo, any lower bound on the complexity of the $\EGLU$ problem also holds for the $\AELU$ one. The $\EXPTIME$-hardness of the two-player $\EGLU$ problem was proved in~\cite{BFLMS08}, the $\PSPACE$-hardness of the one-player version was proved in~\cite{FJ13} (in the equivalent setting of reachability in bounded one-counter automata). Note that those results clearly rely on having an upper bound~$U$ larger than polynomial (w.r.t.~the size of the game) and encoded in binary, as we have already shown that in the opposite case the complexity of the problem is reduced.

Finally, observe that the same reduction and complexities also hold if we use $\AEinf$ instead of $\AEsup$ to define the $\AELU$ problem. This concludes our proof.
\end{proof}

\begin{remark}
One could argue that the reduction from $\AE$ games to $\MPG$ games presented in Lem.~\ref{lem:ae_to_mp} could be used to solve $\AE$ games without resorting to the specific analysis presented in Sect.~\ref{sec:average_games}. Indeed, in the case where the mean-payoff value is zero, any memoryless strategy (which we know to suffice) that is winning should only create zero cycles: the energy can be constrained in the range $[-2\cdot \vert \states\vert \cdot W,\, 2\cdot \vert \states\vert \cdot W]$ along any winning play. However, applying a pseudo-polynomial $\MPG$ algorithm on this new game would only grant $\EXPTIME$-membership for $\AE$ games (because of the polynomial dependency on $W$), in contrast to the $\NP \cap \coNP$ and $\PTIME$ results obtained with the refined analysis for two-player and one-player $\AE$ games respectively.
\end{remark}

\subsection{Memory requirements}

We prove pseudo-polynomial lower and upper bounds on memory for the two players in $\AELU$ games. The upper bound follows from the reduction to a pseudo-polynomial $\AE$ game and the memoryless determinacy of $\AE$ games proved in Thm.~\ref{thm:ae_two_memoryless}. Observe that winning strategies obtained via our reductions have a natural form: they are memoryless w.r.t.~configurations $(s, c)$ denoting the current state and the current energy level. As noted before, when the upper bound on energy $U \in \mathbb{N}$ is polynomial or given in unary, the expanded game is only polynomial in size, and the memory needs are also reduced.

The lower bound can be witnessed in two families of games asking for strategies using memory polynomial in the energy upper bound $U \in \mathbb{N}$ to be won by $\playerOne$ (Fig.~\ref{fig:AELU_memory_p1}) or $\playerTwo$ (Fig.~\ref{fig:AELU_memory_p2}) respectively. It is interesting to observe that those families already ask for such memory when considering the simpler $\EGLU$ objective (i.e., bounded energy only). Sufficiency of pseudo-polynomial memory for $\EGLU$ games follows from~\cite{BFLMS08} but to the best of our knowledge, it was not proved in the literature that such memory is also necessary.

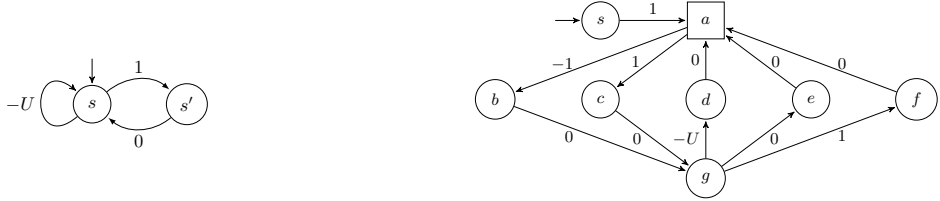
\begin{figure}[htb]
        \centering
\subfloat[$\playerOne$ needs to take $U$ times $(s, s')$ before taking $(s, s)$ once and repeating.]{\label{fig:AELU_memory_p1}\hspace{8mm}\scalebox{0.7}{\begin{tikzpicture}[->,>=stealth',shorten >=1pt,auto,node
    distance=2.5cm,bend angle=45, scale=0.6, font=\normalsize]
    \tikzstyle{p1}=[draw,circle,text centered,minimum size=7mm,text width=4mm]
    \tikzstyle{p2}=[draw,rectangle,text centered,minimum size=7mm,text width=4mm]
    \node[p1]  (0)  at (0, 0) {$s$};
    \node[p1]  (1) at (3, 0) {$s'$};
    \node[]  (empty) at (0, -2.8) {};
    \node[]  (empty) at (-5, 0) {};
    \node[]  (empty) at (5, 0) {};
    
    \coordinate[shift={(0mm,5mm)}] (init) at (0.north);
    \path
    (0) edge [loop left, out=220, in=140,looseness=2, distance=2cm] node [left] {$-U$} (0)
    (init) edge (0);
	\draw[->,>=latex] (0) to[out=40,in=140] node[above] {$1$} (1);
	\draw[->,>=latex] (1) to[out=220,in=320] node[below] {$0$} (0);
      \end{tikzpicture}}\hspace{12mm}}
      \hspace{4mm}
      \subfloat[$\playerTwo$ needs to increase the energy up to $U$ using $(a,c)$ to force $\playerOne$ to take $(g, d)$ then make him lose by taking $(a, b)$.]{\label{fig:AELU_memory_p2}\scalebox{0.7}{\begin{tikzpicture}[->,>=stealth',shorten >=1pt,auto,node
    distance=2.5cm,bend angle=45, scale=1, font=\small]
    \tikzstyle{p1}=[draw,circle,text centered,minimum size=7mm,text width=4mm]
    \tikzstyle{p2}=[draw,rectangle,text centered,minimum size=7mm,text width=4mm]
    \node[p1]  (first)  at (-2, 0) {$s$};
    \node[p2]  (a)  at (0, 0) {$a$};
    \node[p1]  (b) at (-4, -1.5) {$b$};
    \node[p1]  (c) at (-2, -1.5) {$c$};
    \node[p1]  (d) at (0, -1.5) {$d$};
    \node[p1]  (e) at (2, -1.5) {$e$};
    \node[p1]  (f) at (4, -1.5) {$f$};
    \node[p1]  (g) at (0, -3) {$g$};
    \node[]  (empty) at (-6, 0) {};
    \node[]  (empty) at (6, 0) {};
    \coordinate[shift={(-5mm,0mm)}] (init) at (first.west);
    \path
    (first) edge node[above] {$1$} (a)
    (a) edge node[left,xshift=-4mm,yshift=-0.8mm] {$-1$} (b)
    (a) edge node[left,xshift=-1mm] {$1$} (c)
    (d) edge node[left,xshift=0mm] {$0$} (a)
    (e) edge node[right,xshift=1mm] {$0$} (a)
    (f) edge node[right,xshift=4mm,yshift=-0.7mm] {$0$} (a)
    (b) edge node[left,xshift=-4mm,yshift=0.4mm] {$0$} (g)
    (c) edge node[left,xshift=-1mm] {$0$} (g)
    (g) edge node[left,xshift=0mm] {$-U$} (d)
    (g) edge node[right,xshift=1mm] {$0$} (e)
    (g) edge node[right,xshift=4mm,yshift=0.6mm] {$1$} (f)
    (init) edge (first);
      \end{tikzpicture}} }
      \vspace{-1mm}
	\caption{Families of games witnessing the need for pseudo-polynomial-memory strategies for $\EGLU$ (and $\AELU$) objectives. The goal of $\playerOne$ is to keep the energy in $[0,\,U]$ at all times, for $U \in \mathbb{N}$. The left game is won by $\playerOne$ and the right one by $\playerTwo$ but both require memory polynomial in the \textit{value} $U$ to be won.}
	\label{fig:AELU_memory}
\end{figure}

\begin{theorem}
\label{thm:aelu_memory}
Pseudo-polynomial-memory strategies are both sufficient and necessary to win in $\EGLU$ and $\AELU$ games with arbitrary energy upper bound $U \in \mathbb{N}$, for both players. Polynomial memory suffices when~$U$ is polynomial in the size of the game or encoded in unary.
\end{theorem}

\begin{proof}
We first prove the \textit{upper bound on memory}. The expanded game $G'$ built in the reduction from the $\AELU$ to the $\AEG$ problem (Lem.~\ref{lem:aelu_to_ae}) has a state space of size $\vert \states' \vert = (U + 1) \cdot \vert \states \vert + 1$, over which memoryless strategies suffice, by Thm.~\ref{thm:ae_two_memoryless}. Thus, winning for the $\AELU$ objective only requires memory that is polynomial in the original number of states and the upper bound value $U \in \mathbb{N}$. The same reduction holds for $\EGLU$ games with an even simpler \textit{safety} objective (never reaching \textsf{sink}) instead of the $\AE$ one (or equivalently with the $\AE$ objective for threshold $t = U$). Thus, with regard to the \textit{binary encoding} of $U$, strategies require exponential memory in general. For the special cases of unary encoding or polynomially bounded value~$U$, polynomial memory suffices. Note that as usual, these arguments are true for both the $\AEsup$ and the $\AEinf$ versions of the objective.

We now discuss the two families of games witnessing that pseudo-polynomial memory is also a \textit{lower bound} for both players. 

First, consider the one-player game depicted in Fig.~\ref{fig:AELU_memory_p1} and parametrized by the value $U \in \mathbb{N}$. Assume the objective is $\EGLU$, asking for the energy to remain within $[0,\,U]$ at all times. Recall that the initial energy level is fixed to $\initCredit \coloneqq 0$. It is easy to see that there is only one acceptable strategy for $\playerOne$: playing $(s, s')$ exactly $U$ times, then playing the self-loop $(s, s)$ once, and repeating this forever. Indeed, any other strategy eventually leads the energy outside the allowed range. Hence, to win this game, $\playerOne$ needs a strategy described by a Moore machine whose memory contains at least $(U + 1)$ states. This proves that pseudo-polynomial memory is required for $\playerOne$ in $\EGLU$ games. Furthermore, the same argument can be applied on this game with objective $\AELU$ by considering the average-energy threshold $t \coloneqq U$ which is trivially ensured by strategies satisfying the $\EGLU$ objective.

Second, consider the two-player\footnote{In $\EGLU$ games with only $\playerTwo$ (i.e., $S_{1} = \emptyset$), $\playerTwo$ does not need memory to play as he can pick beforehand which of the energy bounds (lower or upper) he will transgress, and then do so with a memoryless strategy.} $\EGLU$ game depicted in Fig.~\ref{fig:AELU_memory_p2}. Again this game is parametrized by the energy upper bound $U \in \mathbb{N}$ and the initial energy level is fixed to $\initCredit \coloneqq 0$. This game can be won by~$\playerTwo$ using the following strategy: if the energy level is in $[1,\, U]$, play $(a,c)$, otherwise play $(a, b)$. Note that this strategy again requires at least $(U + 1)$ states of memory in its Moore machine (to keep track of the energy level). 

This strategy is indeed winning. Observe that $\playerOne$ can only decrease the energy by using edge $(g, d)$ of weight $-U$, and this edge can only be used safely if the energy level is exactly $U$. In addition, the energy is bound to reach or exceed $U$ eventually (as it will increase by $1$ or $2$ between each visit of $a$). If it exceeds $U$, then $\playerTwo$ wins directly. Otherwise, assume that the energy is $U$ when the game is in state~$g$. If $\playerOne$ plays $(g, f)$, he loses (the energy reaches $U+1$). If he plays $(g, e)$, $\playerTwo$ wins by playing $(a,c)$ (the energy also reaches $U+1$). And if $\playerOne$ plays $(g, d)$, $\playerTwo$ wins by playing $(a, b)$ (the energy reaches $-1$). Hence, $\playerTwo$ wins the game against all strategies of $\playerOne$.

Now, observe that $\playerTwo$ \textit{cannot} win if he uses a strategy with less memory states in its Moore machine. Indeed, any such strategy cannot keep track of all the energy levels between $0$ and $U$ and play $(a,c)$ a sufficient number of times in a row before switching to the appropriate choice (depending on the energy being $0$ or $U$). Therefore, if $\playerTwo$ uses such a strategy, $\playerOne$ can maintain the energy in the allowed range by simply reacting to edge $(a, b)$ with $(g, f)$ and to edge $(a, c)$ by choosing between $(g, d)$ (if the energy is $U$) and $(g, e)$ (otherwise). Such choices are safe for $\playerOne$ as the strategy of $\playerTwo$ does not have enough memory to distinguish the resulting energy levels from the intermediate ones.

This proves that $\playerTwo$ also needs pseudo-polynomial memory in $\EGLU$ games. Finally, we remark that this reasoning also holds for the $\AELU$ objective with threshold $t \coloneqq U$, as for the previous game.
\end{proof}

\section{Average-Energy with Lower-Bounded Energy}
\label{sec:average_l}
We conclude with the conjunction of an $\AE$ objective with a lower bound (again equal to zero) constraint on the running energy, but no upper bound. This corresponds to an \textit{hypothetical} unbounded energy storage. Hence, its applicability is limited, but it may prove interesting on the theoretical standpoint.

\begin{bproblem}[$\AEL$]
\label{problem:ael}
Given a game~$\Game$, an initial state $\initState$ and a threshold~$t \in \mathbb{Q}$, decide if $\pI$ has a winning strategy $\St_1 \in \strats_{1}$ for objective $\LBound(\initCredit \coloneqq 0)\,\cap\, \AvgEnergyLevel(t)$.
\end{bproblem}

This problem proves to be challenging to solve: we provide partial answers in the following, with a proper algorithm for one-player games but only a correct but incomplete method for two-player games. As usual, we present our results for the supremum variant $\AEsup$.

\paragraph{Illustration.} Consider the game in Fig.~\ref{fig:aelu_example}. Recall that for $\AELU$ with $U = 3$, the optimal play is $\play_{3}$, and it requires alternation between all three different simple cycles.
Now consider $\AEL$. One may think that relaxing the objective would allow for simpler winning strategies. This is not the case. Some new plays are now acceptable w.r.t.~the energy constraint, such as $\play_{4} = (aabaaba)^{\omega}$, with $\AEsup(\play_{4}) = 11/7$ and $\play_{5} = (aaababa)^{\omega}$, with $\AEsup(\play_{5}) = 18/7$. Yet, the optimal play w.r.t.~the $\AE$ (under the lower-bound energy constraint) is still $\play_{3}$, hence still requires to use all the available cycles, in the appropriate order. This indicates that $\AEL$ games also require complex solutions.

\subsection{One-player games}

We assume that the unique player is $\pI$. Indeed, the opposite case is easy as for $\pII$, the objective is a disjunction and $\pII$ can choose beforehand which sub-objective he will transgress, and do so with a simple memoryless strategy (both $\AEG$ and $\EGL$ games admit memoryless optimal strategies as seen before).
We show that one-player $\AEL$ problems lie in \PSPACE by reduction to $\AELU$ problems for a well-chosen upper bound $U \in \mathbb{N}$ and then application of Thm.~\ref{thm:aelu_reduc}.

\paragraph{The reduction.} Given a game $G = (S_{1}, S_{2} = \emptyset, E, w)$ with largest weight $W \in \mathbb{N}$, an initial state $\initState$,  and a threshold $t \in \mathbb{Q}$, we reduce the $\AEL$ problem to an $\AELU$ problem with an upper bound $U \in \mathbb{N}$ defined as $U \coloneqq t + N^{2} + N^{3}$, with $N = W \cdot (|S|+2)$. Observe that the length of the binary encoding of $U$ is polynomial in the size of the game, the encoding of the largest weight $W$ and the encoding of the threshold~$t$. The intuition is that if $\playerOne$ can win a one-player $\AEL$ game, he can win it without ever reaching energy levels higher than the chosen bound~$U$, even if he is technically allowed to do so. Essentially, the interest of increasing the energy is making more cycles available (as they become safe to take w.r.t.~the lower bound constraint), but increasing the energy further than necessary is not a good idea as it will negatively impact the average-energy. To prove this reduction, we start from an arbitrary winning path in the $\AEL$ game, and build a witness path that is still winning for the $\AEL$ objective, but
also keeps the energy below $U$ at all times. Our construction exploits a
result of Lafourcade et~al. that bounds the value of the counter along a path in a one-counter automaton (stated in~\cite{LLT05} and proved in~\cite[Lem.~42]{LSV:04:16}). We slightly adapt it to our framework in the next lemma. The technique is identical, but the statement is more
precise. In the following, we call an \textit{expanded configuration} of the game $G$ a couple $(s, c)$ where $s \in \states$ is a state and $c \in \mathbb{Z}$ a level of energy. 

\begin{lemma}
  \label{lemma:bound}
  Let $g \in \mathbb{Z}$. Let $(s,c)$ and $(s',c')$ be two expanded
  configurations of the game $G$ such that there exists an expanded
  path $\rho_{\text{exp}} = (s_0,c_0) \dots (s_m,c_m)$ in $G$ from $(s,c)$ to
  $(s',c')$ with $c_i \ge g$ for every $0 \le i \le m$. Then, there is
  a path $\rho_{\text{exp}}'=(s'_0,c'_0) (s'_1,c'_1) \dots (s'_n,c'_n)$ in $G$ from
  $(s,c)$ to $(s',c')$ such that:
  \begin{itemize}
  \item for every $0 \le i \le n$, $g \le c'_i \le \max \{ c,c', g\}
    +N^2+N^3$, where $N = W \cdot (|S|+2)$, with $W$ the maximal
    absolute weight in $G$;
  \item there is an (injective) increasing mapping $\iota \colon
    \{1,\dots,n\} \to \{1,\dots,m\}$ such that for every $1 \le i \le
    n$, $s'_i = s_{\iota(i)}$ and $c'_i \le c_{\iota(i)}$.
  \end{itemize}

  Furthermore, for any two expanded paths $\rho^1$ and~$\rho^2$, with
  $\last(\rho^1) = (s,c)$ and $\first(\rho^2) = (s',c')$, if~$\AE(\rho^1
  \cdot \rho_{\text{exp}} \cdot \rho^2) \le g$, then also $\AE(\rho^1
  \cdot \rho'_{\text{exp}} \cdot \rho^2) \le \AE(\rho^1 \cdot
  \rho_{\text{exp}} \cdot \rho^2) \le g$.
\end{lemma}

\begin{proof}
  We write $\alpha =
  W \cdot (|S|+1)$, $\beta= (\alpha+W) \cdot (\alpha+W-1)-1$ and $K =
  \max\bigl\lbrace c,c', g\bigr\rbrace+(\alpha+W)^2$. We~apply inductively a
  transformation that removes similar ascending and descending segments of the
  path. The~segments are selected such that their composition is neutral
  w.r.t.~the energy. 

  Pick a subpath $\rho_{\text{exp}}[k,k+h] = (s_k, c_k)\ldots{}(s_{k+h}, c_{k+h})$ of $\rho_{\text{exp}}$, if it exists, such that:
  \begin{enumerate}[\hspace{3mm}$(a)$]
  \item $c_{k} \le K$ and $c_{k+h} \le K$;
  \item for every $0 < \ell < h$, $c_{k+\ell} > K$;
  \item there is $0<\ell<h$ such that $c_{k+\ell} > K + W \cdot
    (|S|+1) \cdot \beta$.
  \end{enumerate}
  If such a subpath does not exist, then this means that the cost
  along $\rho_{\text{exp}}$ is overall bounded by $K+W\cdot (|S|+1)\cdot \beta$ (since
  condition $(a)$ is not restrictive\,---\,$c,c' \le K$), which then
  concludes the proof. Hence, assume such a subpath exists for the following steps.

  \paragraph{Ascent part.}  Let $k \le \ell_0\le\dots \le
  \ell_{\beta} \le k+h$ be indices such that:
  \begin{itemize}
  \item $c_{\ell_i} > K + i \cdot W \cdot (|S|+1)$;
  \item for every $k \le \ell< \ell_i$, $c_\ell \le K + i \cdot W
    \cdot (|S|+1)$.
  \end{itemize}
  Fix $0 \le i \le \beta$. Then it holds that $c_{\ell_i} \le K + i
  \cdot W \cdot (|S|+1) + W$ and thus $c_{\ell_{i+1}} - c_{\ell_i} > K
  + (i+1) \cdot W \cdot (|S|+1) -(K + i \cdot W \cdot (|S|+1) + W) = W
  \cdot (|S|+1) -W = W \cdot |S|$.  Let $J_i$ be a subset of
  $[\ell_i;\ell_{i+1}]$ defined by $\ell_i \in J_i$, and if $j \in
  J_i$, then let $j' \le \ell_{i+1}$ be the smallest index larger than
  $j$ (if it exists) such that $c_{j'}>c_j$. Obviously we have $c_j <
  c_{j'} \le c_j+W$. Hence the cardinal of $J_i$ is at least $1+
  \frac{W \cdot |S|}{W} \ge |S|+1$. Hence there is a state
  $\widetilde{s}^{(i)}$ and two indices $j_{i,1}<j_{i,2} \in J_i$ with
  $(s_{j_{i,1}},c_{j_{i,1}}) =(\widetilde{s}^{(i)},\alpha_1)$ and
  $(s_{j_{i,2}},c_{j_{i,2}}) =(\widetilde{s}^{(i)},\alpha_2)$ with
  $c_{\ell_i} \le \alpha_1 < \alpha_2 \le c_{\ell_{i+1}}$, hence
  using previous computed bounds, $0< \alpha_2-\alpha_1 \le
  c_{\ell_{i+1}} - c_{\ell_i} < W \cdot (|S|+2) = \alpha+W$.  We
  write $\widetilde{d}^{(i)} = \alpha_2-\alpha_1$. The segment between
  indices $j_{i,1}$ and $j_{i,2}$ is a candidate for being removed.
  Due to the value of $\beta$, there is $d \in \{\widetilde{d}^{(i)}
  \mid 0 \le i \le \beta\}$ that appears $(\alpha+W)$ times in that
  set.

  \paragraph{Descent part.} We do a similar reasoning for the
  ``descent'' part. There must exist indices $k \le m_0 \le \dots \le
  m_\beta \le k+h$ such that:
  \begin{itemize}
  \item $c_{m_i} > K + (\beta-i) \cdot W \cdot (|S|+1)$;
  \item for every $m_i<m \le k+h$, $c_m \le K + (\beta-i) \cdot W
    \cdot (|S|+1)$.
  \end{itemize}
  Note that we obviously have $\ell_\beta<m_0$.

  Then we apply the same combinatorics as for the ascent part. There
  is some value $0<d' <\alpha+W$ which appears at least $\alpha+W$
  times in potential cycles within the segment $\rho_{\text{exp}}[k,k+h]$.

  \paragraph{Transformation.} The algorithm then proceeds by removing $d'$ segments that increase
  the cost by $d$ within $\rho_{\text{exp}}[\ell_0,\ell_\beta]$ and $d$ segments
  that decrease the cost by $d'$ within $\rho_{\text{exp}}[m_0,m_\beta]$.  This
  yields another path $\rho_{\text{exp}}'$ and an obvious injection of $\rho_{\text{exp}}'$ into
  $\rho_{\text{exp}}$ which satisfies all the mentioned constraints.  The sum of
  all energy levels along $\rho_{\text{exp}}'$ is smaller than that along $\rho_{\text{exp}}$,
  and any energy level along $\rho_{\text{exp}}'$ is obtained from that along
  $\rho_{\text{exp}}$ by decreasing by at most $0<d \cdot d' < (\alpha+W)^2$.  By
  assumption on segment $\rho_{\text{exp}}[k,k+h]$ and bound $K$, we get that
  the cost along $\rho_{\text{exp}}'$ is always larger than or equal to $g$, $c$
  and $c'$.

  We iterate this transformation to get a uniform upper bound.  We
  finally notice that the obtained upper bound $K + W \cdot (|S|+1)
  \cdot \beta$ is bounded itself by $\max \{c,c',g\}+N^2+N^3$, where $N =
  W \cdot (|S|+2)$. This implies the expected result.
\end{proof}

We build upon this lemma to define an appropriate transformation leading to the witness path and derive a sufficiently large upper bound $U \in \mathbb{N}$ for the $\AELU$ problem.

\begin{lemma}
\label{lem:ael_to_aelu}
The $\AEL$ problem over a one-player game $G = (S_{1}, S_{2} = \emptyset, E, w)$, with an initial state $\initState$ and a threshold $t \in \mathbb{Q}$, is reducible to an $\AELU$ problem over the same game $G$, for the same threshold~$t$ and upper bound $U \coloneqq t + N^{2} + N^{3}$, with $N = W \cdot (|S|+2)$.
\end{lemma}

\begin{proof}
  We prove that we can bound the energy along a
  witness of the one-player $\AEL$ problem. Let $\sigma$ be a winning strategy of $\pI$ for the objective
  $\LBound(\initCredit \coloneqq 0) \cap \AvgEnergyLevel(t)$ and $\pi = s_0 s_1 \dots s_n \dots$ be the
  corresponding outcome.

  We build another strategy $\widetilde\sigma$ with corresponding
  play $\widetilde\pi$ such that for every $n$, $0 \le
  \initCredit+\EL(\widetilde\play(n)) \le \initCredit + t + N^2 +
  N^3$, where $N = W \cdot (|S|+2)$ ($W$ is the maximal absolute
  weight in $G$), and such that $\AEsup(\widetilde\pi) \le \AEsup(\pi)$. We
  actually build the play $\widetilde\pi$ directly, and infer strategy
  $\widetilde\sigma$.

  From $\pi$, we build the expanded play $\play_{\text{exp}} = (s_0,c_0) (s_1,c_1)
  \dots (s_n,c_n) \dots$ such that $c_i = \EL(\play(i))$ for every $i
  \ge 0$. Since $\pi$ is a witness satisfying the objective
  $\LBound(\initCredit) \cap \AvgEnergyLevel(t)$, it holds that $c_i +
  \initCredit \ge 0$ for every $i \ge 0$. We now show that some pair
  $(s,c)$ is visited infinitely often along $\play_{\text{exp}}$. Toward a
  contradiction, assume that it is not the case. Then since energy levels are bounded from below along $\play$, this means that
  $\liminf_{n\to \infty} c_n = \TPinf(\play) = +\infty$, and by Lem.~\ref{lem:AEbetweenTP}, that $\AEsup(\pi) =
  +\infty$ which contradicts the play being winning for the $\AE$ objective with threshold $t \in \mathbb{Q}$. Now select the smallest energy $c$ and state $s$ such that
  $(s,c)$ is visited infinitely often along $\play_{\text{exp}}$. Pick $n_0$ such
  that (1)~$(s_{n_0}, c_{n_0}) = (s,c)$, (2)~$\pi[\ge n_0] = s_{n_0}s_{n_0 + 1}\ldots{}$ only visits
  states that are visited infinitely often along $\pi$, and (3)~for
  every $(s',c')$ along $\play_{\text{exp}}[\ge n_0]$, it holds that $c' \ge c$.
  
  We can then write $\play_{\text{exp}}$ as $\play_{\text{exp}}[\le n_0] \cdot \mathcal{C}_1
  \cdot \mathcal{C}_2 \dots$ where each $\mathcal{C}_i$ ends at
  configuration $(s,c)$ (hence $\mathcal{C}_i$ forms a cycle), and
  each configuration $(s',c')$ along some $\mathcal{C}_i$ satisfies
  $c' \ge c$. We write $\gamma_i$ for the projection of
  $\mathcal{C}_i$ on states (without energy level)\,---\,it forms a cycle as
  well. We obviously have $\AEsup(\pi) = \EL(\pi(n_0)) + \AEsup(\pi[>
    n_0]) = c + \AEsup(\pi[> n_0])$ by Lem.~\ref{lem:AE_prefix}, and since $\AEsup(\pi) \le t$,
  there must be some cycle $\mathcal{C}_i$ such that $\AE(\gamma_i)
  \le t - c$. We write $\gamma$ for such a $\gamma_i$, and we define $\varpi =
  \pi(n_0) \cdot \gamma^\omega$: it is a lasso-shaped play which also satisfies the objective $\LBound(\initCredit) \cap
  \AvgEnergyLevel(t)$.

  We will now modify the play $\varpi$, so that the energy does
  not grow too much along it. We write $\varpi_{\text{exp}}$ for the expanded version of
  $\varpi$: it is of the form $\varpi_{\text{exp}}[\le n_0] \cdot
  \bigl(\varpi_{\text{exp}}[n_0+1,n_0+p]\bigr)^\omega$, where
  $\varpi_{\text{exp}}[n_0+1,n_0+p]$ projects onto $\gamma$ when the energy information
  is removed (note that the last configurations of $\varpi_{\text{exp}}[\le n_0]$
  and of $\varpi_{\text{exp}}[n_0+1,n_0+p]$ are $(s,c)$). We will do two things:
  $(i)$ first we will work on the cycle $\gamma$; and $(ii)$ then we
  will work on the prefix $\varpi[\le n_0]$, to build a witness with a
  fixed upper bound on the energy. For the rest of the proof, we assume
  that $\varpi_{\text{exp}} = (s_0,c_0) (s_1,c_1) \dots$ so that $(s_{n},c_{n}) =
  (s,c)$ for every $n = n_0 + b \cdot p$ for some integer~$b$.

First consider point $(i)$. Let us notice that $c \le t$, otherwise the average-energy
    along $\varpi$ could not be at most $t$ (remember
    that the cost along the expanded version of $\gamma$ starting at
    $(s,c)$ is always larger than or equal to $c$ by construction).  We pick the first
    maximal subpath $\varpi_{\text{exp}}[k,k+h]$ of $\varpi_{\text{exp}}$ with $[k,k+h]
    \subseteq (n_0,n_0+p)$, such that $c_{k+\ell} >t$ for every $0 \le
    \ell \le h$. By maximality of $\varpi_{\text{exp}}[k,k+h]$, it is the case
    that $c_{k-1} \le t$ and $c_{k+h+1} \le t$. We infer that $t<c_{k}
    \le t+W$ and $t<c_{k+h} \le t+W$, where $W$ is the maximal
    absolute weight in the game $G$.  We
    apply~Lem.~\ref{lemma:bound} to the path
    $\varpi_{\text{exp}}[k,k+h]$ with $g = t$, and we get that we can build an
    expanded path $\varpi_{\text{exp}}^{(k)}$ which is shorter than
    $\varpi_{\text{exp}}[k,k+h]$ and such that:
    \begin{itemize}
    \item at all positions of $\varpi_{\text{exp}}^{(k)}$, the energy is in the interval $[t,t+N^2+N^3]$, where $N = W \cdot (|S|+2)$;
    \item there is an injective increasing mapping $\iota \colon
      [0,|\varpi_{\text{exp}}^{(k)}|] \to [k,k+h]$ such that for every index $1 \le i
      \le |\varpi_{\text{exp}}^{(k)}|$, the state of $\varpi_{\text{exp}}^{(k)}[=i]$
      coincides with that of $\varpi_{\text{exp}}[=\iota(i)]$ and the energy at
      position $i$ of $\varpi_{\text{exp}}^{(k)}$ is smaller than or equal to
      $c_{\iota(i)}$.
    \end{itemize}
    In particular, we have a new witness for the objective
    $\LBound(\initCredit) \cap \AvgEnergyLevel(t)$, which is the play
    $\varpi[< n_0] \cdot \bigl(\varpi[n_0,k-1] \cdot \varpi^{(k)}
    \cdot \varpi[k+h+1,n_0+|\gamma|-1]\bigr)^\omega$, where
    $\varpi^{(k)}$ is the projection of $\varpi_{\text{exp}}^{(k)}$ over the
    states of the game $G$. We iterate this transformation over all relevant segments
    of $\gamma$ (this will happen only a finite number of times), and
    we end up with a new lasso-play $\varpi' = \varpi[\le n_0] \cdot
    (\gamma')^\omega$ such that:
    \begin{itemize}
    \item $\varpi'$ satisfies the objective $\LBound(\initCredit) \cap
      \AvgEnergyLevel(t)$;
    \item for every $1\le \ell \le |\gamma'|$, $-\initCredit \le
      \EL(\varpi'(n_0+\ell)) \le t+N^2+N^3$.
    \end{itemize}

Now, consider point $(ii)$. It remains to work on the prefix $\varpi[\le n_0]$ (which is
    still a prefix of $\varpi'$). We apply Lem.~\ref{lemma:bound} to
    the prefix $\varpi[\le n_0]$ with $g=0$, and we get an appropriately
    bounded witness.
    
Summing up, our construction proves that if there exists a winning play for $\LBound(\initCredit \coloneqq 0) \cap \AvgEnergyLevel(t)$ in the one-player game $G$, then there exists one for $\LUBound(U, \initCredit \coloneqq 0) \cap \AvgEnergyLevel(t)$, with $U \coloneqq t + N^2 + N^3$. Since the converse implication is obvious (as the second objective is strictly stronger), this concludes the proof of the reduction to an $\AELU$ game.
\end{proof}

\paragraph{Complexity.} Plugging this bound $U$ in the \PSPACE algorithm for one-player $\AELU$ games (Thm.~\ref{thm:aelu_reduc}) implies \PSPACE-membership for one-player $\AEL$ games also. In terms of time complexity, we saw that this problem can thus be solved in pseudo-polynomial time. We prove that no truly-polynomial-time algorithm can be obtained unless $\PTIME = \NP$ as the one-player $\AEL$ problem is $\NP$-hard. We show it by reduction from the subset-sum problem~\cite{garey_FNY1979}: given a finite set of naturals $A = \{a_{1}, \ldots{}, a_{n}\}$ and a target natural $v$, decide if there exists a subset $B \subseteq A$ such that $\sum_{a_{i} \in B} a_{i} = v$. The reduction is sketched in Fig.~\ref{fig:subset-sum_to_ael}: a play corresponds to a choice of subset. In order to keep a positive energy level, $\playerOne$ has to pick a subset that achieves a sum \textit{at least} equal to $v$, but in order to satisfy the $\AE$ threshold, this sum must be \textit{at most} $v$: hence $\playerOne$ must be able to pick a subset whose sum is \textit{exactly} the target $v$.

\vspace{-4mm}
\begin{figure}[thb]
        \centering
\scalebox{0.9}{\begin{tikzpicture}[->,>=stealth',shorten >=1pt,auto,node
    distance=2.5cm,bend angle=45, scale=0.95, font=\normalsize]
    \tikzstyle{p1}=[draw,circle,text centered,minimum size=7mm,text width=8mm]
    \tikzstyle{p2}=[draw,rectangle,text centered,minimum size=7mm,text width=4mm]
    \node[p1]  (s1)  at (0, 0) {$s_{1}$};
    \node[p1]  (a1) at (2, 1.2) {$a_{1}$};
    \node[p1]  (na1) at (2, -1.2) {$\neg a_{1}$};
    \node[p1]  (s2)  at (4, 0) {$s_{2}$};
    \node[p1]  (a2) at (6, 1.2) {$a_{2}$};
    \node[p1]  (na2) at (6, -1.2) {$\neg a_{2}$};
    \node[]  (s3)  at (8, 0) {};
    \node[p1]  (sn)  at (10, 0) {$s_{n}$};
    \node[p1]  (an) at (12, 1.2) {$a_{n}$};
    \node[p1]  (nan) at (12, -1.2) {$\neg a_{n}$};
    \node[p1]  (end)  at (14, 0) {$\textsf{end}$};
    
    \coordinate[shift={(-8mm,0mm)}] (init) at (0.west);
    \path
    (end) edge [loop left, out=40, in=320,looseness=2, distance=2cm] node [right] {$0$} (end)
    (init) edge (s1)
    (s1) edge node[above] {$a_{1}$} (a1)
    (s1) edge node[below] {$0$} (na1)
    (a1) edge node[above] {$0$} (s2)
    (na1) edge node[below] {$0$} (s2)
    (s2) edge node[above] {$a_{2}$} (a2)
    (s2) edge node[below] {$0$} (na2)
    (a2) edge node[above] {$0$} (s3)
    (na2) edge node[below] {$0$} (s3)
    (sn) edge node[above] {$a_{n}$} (an)
    (sn) edge node[below] {$0$} (nan)
    (an) edge node[above] {$-v$} (end)
    (nan) edge node[below] {$-v$} (end)
    (s3) edge[loosely dotted,-,thick] (sn);
      \end{tikzpicture}}
	\caption{Reduction from the subset-sum problem for target $v \in \mathbb{N}$ to a one-player $\AEL$ problem for average-energy threshold $t \coloneqq v$.}
	\label{fig:subset-sum_to_ael}
\end{figure}
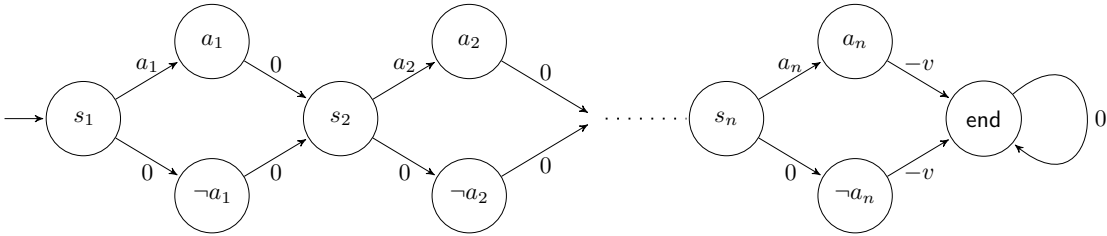
\vspace{-4mm}

\begin{theorem}
\label{thm:ael_one_complexity}
The $\AEL$ problem is in $\PSPACE$ and at least $\NP$-hard for one-player games.
\end{theorem}

\begin{proof}
First, consider the claim of $\PSPACE$-membership. Let $G= (S_{1}, S_{2} = \emptyset, E, w)$ be a game with initial state $\initState$. Consider the $\AEL$ problem for a given average-energy threshold $t \in \mathbb{Q}$. By Lem.~\ref{lem:ael_to_aelu}, this problem is reducible to the $\AELU$ problem with upper bound $U \coloneqq t + N^{2} + N^{3}$, with $N = W \cdot (|S|+2)$. Hence, $U$ is of order $\mathcal{O}(t + W^3 \cdot \vert S\vert^{3})$, and its encoding is polynomial in the encoding of the original $\AEL$ problem (including thresholds and weights, not only in the number of states of the original game!). Following the complexity analysis presented in Thm.~\ref{thm:aelu_reduc}, we thus conclude that the one-player $\AEL$ problem is indeed in $\PSPACE$. In terms of time, by using the $\MP$ reduction and the pseudo-polynomial algorithm, we have an algorithm for the one-player $\AEL$ problem that takes time of order
\begin{equation*}
\mathcal{O}\left( \big( (U + 1) \cdot \vert \states \vert + 1\big) ^{3} \cdot \max \{U,\, \lceil t\rceil + 1\}\right) =  \mathcal{O}\left(\left( t + W^3 \cdot |S|^{3}\right)^{4} \cdot |S|^{3}\right),
\end{equation*}
which is still pseudo-polynomial in the size of the original $\AEL$ problem (i.e., polynomial in the number of states and in the values of the largest absolute weight and of the average-energy threshold).

Second, we prove that the one-player $\AEL$ problem is $\NP$-hard. Consider the subset-sum problem for the set $A = \{a_1, \ldots{}, a_{n}\}$ such that for all $i \in \{1, \ldots{}, n\}$, $a_{i} \in \mathbb{N}$, and target $v \in \mathbb{N}$. Deciding if there exists a subset $B \subseteq A$ such that $\sum_{a_{i} \in B} a_{i} = v$ is well-known to be $\NP$-complete~\cite{garey_FNY1979}. We reduce this problem to an $\AEL$ problem over the game $G$ depicted in Fig.~\ref{fig:subset-sum_to_ael}. Observe that this game has polynomially as many states as the size of $A$, and that its largest absolute weight is equal to the maximum between the largest element of~$A$ and the target $v$. It is clear that there is a bijection between choices of subsets of $A$ and plays in $G$. Let us fix threshold $t \coloneqq v$ for the average-energy. Recall that Lem.~\ref{lem:AE_prefix} implies that the average-energy of any play is exactly its energy level at the first visit of \textsf{end} (because afterwards the zero self-loop is repeated forever). Hence, we have that
\begin{enumerate}
\item a play $\play$ in $G$ is winning for $\LBound(\initCredit \coloneqq 0)$ if and only if the corresponding subset $B$ is such that $\sum_{a_{i} \in B} a_{i} \geq v$;
\item a play $\play$ in $G$ is winning for $\AvgEnergyLevel(t \coloneqq v)$ if and only if the corresponding subset $B$ is such that $\sum_{a_{i} \in B} a_{i} \leq v$.
\end{enumerate}
Therefore, $\playerOne$ has a winning strategy for the $\AEL$ objective $\LBound(\initCredit \coloneqq 0) \cap \AvgEnergyLevel(t \coloneqq v)$ in $G$ if and only if there exists a subset $B$ for which the sum of elements is exactly equal to the target $v$.

This proves the reduction from the subset-sum problem and the $\NP$-hardness result. Observe two things. First, the hardness proof relies on having set elements and a target value that are not polynomial in the size of the input set $A$. Indeed, the subset-sum problem is solvable with a pseudo-polynomial algorithm, hence in~$\PTIME$ for polynomial values. Second, our reduction also holds for the $\AEinf$ variant of the average-energy.
\end{proof}

\paragraph{Memory requirements.} Recall that for $\playerTwo$, the situation is simpler and memoryless strategies suffice. By the reduction to $\AELU$, we know that pseudo-polynomial memory suffices for $\playerOne$. This bound is tight as witnessed by the family of games already presented in Fig.~\ref{fig:AELU_memory_p1}. To ensure the lower bound on energy, $\playerOne$ has to play edge $(s, s')$ at least $U$ times before taking the $(s, s)$ self-loop. But to minimize the average-energy, edge $(s, s')$ should never be played more than necessary. The optimal strategy is the same as for the $\AELU$ problem: playing $(s, s')$ exactly $U$ times, then $(s, s)$ once, then repeating, forever. As shown in Thm.~\ref{thm:aelu_memory}, this strategy requires pseudo-polynomial memory.

\begin{theorem}
\label{thm:ael_one_memory}
Pseudo-polynomial-memory strategies are both sufficient and necessary to win for $\playerOne$ in one-player $\AEL$ games. Memoryless strategies suffice for $\playerTwo$ in such games.
\end{theorem}

\subsection{Two-player games}

For the two-player $\AEL$ problem, we only provide partial answers, as open questions remain. We first discuss decidability: we present an incremental algorithm that is correct but incomplete (Lem.~\ref{lem:ael_semi}) and we draw the outline of a potential approach to obtain completeness hence decidability. Then, we prove that the two-player $\AEL$ problem is at least $\EXPTIME$-hard (Lem.~\ref{lem:ael_exp_hard}). Finally, we show that, in contrast to the one-player case, $\playerTwo$ also requires memory in two-player $\AEL$ games (Lem.~\ref{lem:ael_memory}).

\paragraph{Decidability.} Assume that there exists some $U \in \mathbb{N}$ such that $\playerOne$ has a winning strategy for the $\AELU$ problem with upper bound $U$ and average-energy threshold $t$. Then, this strategy is trivially winning for the $\AEL$ problem as well. This observation leads to an incremental algorithm that is correct (no false positives) but incomplete (it is not guaranteed to stop).

\begin{lemma}
\label{lem:ael_semi}
There is an algorithm that takes as input an $\AEL$ problem and iteratively solves corresponding $\AELU$ problems for incremental values of $U \in \mathbb{N}$. If a winning strategy is found for some $U \in \mathbb{N}$, then it is also winning for the original $\AEL$ problem. If no strategy is found up to value $U \in \mathbb{N}$, then no strategy of~$\playerOne$ can simultaneously win the $\AEL$ problem and prevent the energy from exceeding~$U$ at all times.
\end{lemma}

While an incomplete algorithm clearly seems limiting from a theoretical standpoint, it~is worth noting that in practice, such approaches are common and often necessary restrictions, even for problems where a complete algorithm is known to exist. For example, the existence of an initial energy level sufficient to win in multi-dimensional energy games can be decided~\cite{CRR14} but practical implementations resort to an incremental scheme that is in practice incomplete because the theoretical bound granting completeness is too large to be tackled efficiently by software synthesis tools~\cite{BBFR13}. In our case, we have already seen that if such a bound exists for the two-player $\AEL$ problem, it needs to be at least exponential in the encoding of problem (cf.~one-player $\AEL$ games). Hence it seems likely that a prohibitive bound would be necessary, rendering the algorithm of Lem.~\ref{lem:ael_semi} more appealing in practice.

Nevertheless, we \textit{conjecture} that the $\AEL$ problem is decidable for two-player games and that, similarly to the one-player case, an upper bound on the energy can be obtained. Unfortunately, this claim is much more challenging to prove for two-player games. Clearly, the approach of Lem.~\ref{lem:ael_to_aelu} has to be generalized: while in one-player games we could pick a witness winning play and transform it, we now have to deal with \textit{tree unfoldings}\,---\,describing sets of plays\,---\,because of the uncontrollable choices made by $\playerTwo$.

A potentially promising approach is to define a notion close to the \textit{self-covering trees} used in~\cite{CRR14} for energy games. Roughly, take any winning strategy of $\playerOne$ in a two-player $\AEL$ game. Without further assumption, this strategy could be infinite-memory. It can be represented by its corresponding infinite tree unfolding where in nodes of $\playerOne$, a unique child is given by the strategy, and in nodes of $\playerTwo$, all possible successors yield different branches. Every rooted branch of this tree is infinite and describes a winning play. Then, we would like to achieve the following steps.
\begin{enumerate}
\item\label{item:cut} Prove that all branches of this unfolding can be cut in such a way that the resulting finite tree describes a \textit{finite-memory} strategy that is still winning for the $\AEL$ objective.
\item Reduce the height of this finite tree by compressing parts of the branches: deleting embedded zero cycles seems to be a good candidate for the transformation to apply.
\item\label{item:bound} Derive an \textit{upper bound} on the height of the compressed tree and, consequently, on the maximal energy level reached along any play consistent with the corresponding strategy.
\item Use this upper bound to reduce the $\AEL$ problem to an $\AELU$ problem.
\end{enumerate}
Sadly, some challenges appear on the technical side when trying to implement this approach, mainly for items~\ref{item:cut} and~\ref{item:bound}. Intuitively, the additional difficulty (when compared to the approach developed in~\cite{CRR14} and similar works) arises from the fact that describing what is a good cycle pattern for the $\AEL$ objective is much more intricate than it is for a simple $\EGL$ objective (in which case we simply look for zero cycles). This makes the precise definition of an appropriate transformation of branches, and the resulting tree height analysis, more tedious to achieve.

We also mention that the $\AEL$ problem could be reduced, following a construction similar to the one given in Sect.~\ref{subsec:AELU_algo}, to a mean-payoff threshold problem over an infinite arena, where states of the expanded graph are arranged respectively to their energy level, ranging from zero to infinity, and where weights would also take values inside $\mathbb{N} \cup \{\infty\}$ (as they reflect the possible energy levels). To the best of our knowledge, it is not known if mean-payoff games over such particular structures are decidable. If so, an algorithm would have to fully exploit the peculiar form of those arenas, as it is for example known that general models such as pushdown games are undecidable for the mean-payoff~\cite{CV12}.

Finally, one could envision to fill the gap between one-player and two-player $\AEL$ games by using a general result similar to~\cite[Cor.~7]{GZ05}. Recall that we used it to derive memoryless determinacy in the two-player case from memoryless determinacy of both one-player versions ($S_{1} = \emptyset$ and $S_{2} = \emptyset$). However, we here have that in one-player games, $\playerOne$ requires pseudo-polynomial memory. Therefore, it is necessary to extend the result of Gimbert and Zielonka to finite-memory strategies: that is, to show that if we have a bound on memory valid in both one-player versions of a game, then this bound, or a derived one, is also valid in the two-player version. This is not known to be the case in general, and establishing it for a sufficiently general class of games seems challenging.

\paragraph{Complexity lower bound.} We now prove that the two-player $\AEL$ problem would require at least exponential time to solve. Our proof is by reduction from \textit{countdown games}. A~countdown game
$\mathcal{C}$ is a weighted graph $(\mathcal{V}, \mathcal{E})$, where $\mathcal{V}$ is the finite set of
states, and $\mathcal{E} \subseteq \mathcal{V} \times \mathbb{N} \setminus \{0\} \times \mathcal{V}$ is the edge relation. Configurations are of the form $(v, c)$, $v \in \mathcal{V}$, $c \in \mathbb{N}$. The game starts in an initial configuration $(v_{\text{init}}, c_0)$ and transitions from a configuration $(s, c)$ are performed as follows. First, $\pI$ chooses a duration $d$, $0 < d \leq c$ such that there exists $e = (v, d, v') \in \mathcal{E}$ for some $v' \in \mathcal{V}$. Second, $\pII$ chooses a state $v' \in \mathcal{V}$ such that $e = (v, d, v') \in \mathcal{E}$. Then the game advances to $(v', c-d)$. Terminal configurations are reached whenever no legitimate move is available. If such a configuration is of the form $(v, 0)$, $\pI$ wins the play, otherwise $\pII$ wins. Deciding the winner given an initial configuration $(v_{\text{init}}, c_0)$ is $\EXPTIME$-complete~\cite{JSL08}.

Our reduction is depicted in Fig.~\ref{fig:countdown_to_ael}. The $\EL$ is initialized to $c_0$, then it is decreasing along any play. Consider the $\AEL$ objective for $\AE$ threshold $t \coloneqq 0$. To ensure that the energy always stays non-negative, $\playerOne$ has to switch to \textsf{stop} while the $\EL$ is no less than zero. In addition, to ensure an $\AE$ no more than $t = 0$, $\playerOne$ has to obtain an $\EL$ at most equal to zero before switching to \textsf{stop} (as the $\AE$ will be equal to this $\EL$ thanks to Lem.~\ref{lem:AE_prefix} and the zero self-loop on \textsf{stop}). Hence, $\playerOne$ wins the $\AEL$ objective only if he can ensure a total sum of chosen durations that is \textit{exactly} equal to $c_{0}$, i.e., if he can reach a winning terminal configuration for the countdown game. The converse also holds.

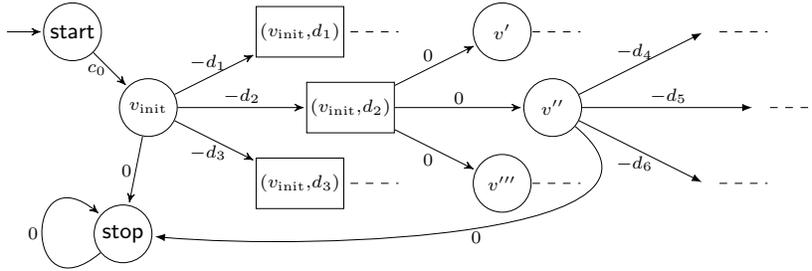
\begin{figure}[htb]
        \centering
\scalebox{0.96}{\begin{tikzpicture}[->,>=stealth',shorten >=1pt,auto,node
    distance=2.5cm,bend angle=45,font=\normalsize,scale=.7,inner sep=.5mm]
    \everymath{\scriptstyle}
    \tikzstyle{p1}=[draw,circle,text centered,minimum size=8mm]
    \tikzstyle{p2}=[draw,rectangle,text centered,minimum size=7mm]
    \node[p1]  (start)  at (-1.5, 1.5) {\small\textsf{start}};
    \node[p1]  (1)  at (0, 0) {$v_{\text{init}}$};
    \node[p2]  (2a) at (3, 1.5) {$(v_{\text{init}}, d_1)$};
    \node[p2]  (2) at (4, 0) {$(v_{\text{init}}, d_2)$};
    \node[p2]  (2b) at (3, -1.5) {$(v_{\text{init}}, d_3)$};
    \node[p1]  (3)  at (8, 0) {$v''$};
    \node[p1]  (3a)  at (7, 1.5) {$v'$};
    \node[p1]  (3b)  at (7, -1.5) {$v'''$};
    \node[p1]  (4) at (-0.5, -2.5) {\small\textsf{stop}};
    \path[use as bounding box] (12,0);
    \coordinate[shift={(-5mm,0mm)}] (init) at (start.west);
    \path
    (init) edge (start)
    (start) edge node[left] {$c_{0}$} (1)
    (1) edge node [left] {$0$} (4)
    (1) edge node [above] {$-d_2$} (2)
    (1) edge node [above,xshift=-1mm] {$-d_1$} (2a)
    (1) edge node [below,xshift=-1mm] {$-d_3$} (2b)
    (2) edge node [above] {$0$} (3)
    (2) edge node [above,xshift=-1mm] {$0$} (3a)
    (2) edge node [below,xshift=-1mm] {$0$} (3b)
    (4) edge [loop left, out=220, in=140,looseness=2, distance=2cm] node [left] {$0$} (4)
    ;
	\draw[->,>=latex] (3) to[out=320, in=355] node [below] {$0$} (4);
	\draw[->,>=latex] (3) to node[above,xshift=-1mm] {$-d_4$} (11,1.5);
	\draw[->,>=latex] (3) to node[above] {$-d_5$} (12,0);
	\draw[->,>=latex] (3) to node[below,xshift=-1mm] {$-d_6$} (11,-1.5);
	\draw[dashed,-,>=latex] (12.3, 0) to (13.3,0);
	\draw[dashed,-,>=latex] (11.3, 1.5) to (12.3,1.5);
	\draw[dashed,-,>=latex] (11.3, -1.5) to (12.3,-1.5);
	\draw[dashed,-,>=latex] (4, 1.5) to (5,1.5);
	\draw[dashed,-,>=latex] (4, -1.5) to (5,-1.5);
	\draw[dashed,-,>=latex] (7.6, 1.5) to (8.6,1.5);
	\draw[dashed,-,>=latex] (7.6, -1.5) to (8.6,-1.5);
\end{tikzpicture}}
\vspace{-1mm}
	\caption{Reduction from a countdown game $\mathcal{C} = (\mathcal{V}, \mathcal{E})$ with initial configuration $(v_{\text{init}}, c_{0})$ to a two-player $\AEL$ problem for average-energy threshold $t \coloneqq 0$.}
	\label{fig:countdown_to_ael}
\end{figure}

\begin{lemma}
\label{lem:ael_exp_hard}
The $\AEL$ problem is $\EXPTIME$-hard for two-player games.
\end{lemma}

\begin{proof}
Given a countdown game $\mathcal{C} = (\mathcal{V}, \mathcal{E})$ and an initial configuration $(v_{\text{init}}, c_0)$, we build a game $G = (S_1, S_2, E, w)$ with initial state $s_{\text{init}}$ such that $\playerOne$ has a winning strategy in $G$ for the $\AEL$ objective for threshold $t \coloneqq 0$ if and only if he has a winning strategy in $\mathcal{C}$ to reach a terminal configuration with counter value zero. The construction is depicted in Fig.~\ref{fig:countdown_to_ael}. Formally, the game $G$ is built as follows.
\begin{itemize}
\item $S_1 = \mathcal{V} \cup \{\textsf{start}, \textsf{stop}\}$.
\item $S_2 = \left\lbrace (v, d) \in \mathcal{V} \times \mathbb{N} \setminus \{0\} \mid \exists\, v' \in \mathcal{V},\, (v, d, v') \in \mathcal{E}\right\rbrace$.
\item $s_{\text{init}} = \textsf{start}$.
\item For each $(v, d, v') \in \mathcal{E}$, we have that $(v, (v, d)) \in E$ with $w(v, (v, d)) = -d$ and $((v, d), v') \in E$ with $w((v, d), v') = 0$.
\item Additionally, $(\textsf{start}, v_{\text{init}}) \in E$ with $w(\textsf{start}, v_{\text{init}}) = c_{0}$, $(\textsf{stop}, \textsf{stop}) \in E$ with $w(\textsf{stop}, \textsf{stop}) = 0$ and for all $v \in \mathcal{V}$, $(v, \textsf{stop}) \in E$ with $w(v, \textsf{stop}) = 0$.
\end{itemize}

First, consider the left-to-right direction of the claim. Assume $\playerOne$ has a winning strategy for the $\AEL$ objective in $G$. As noted before, such a strategy necessarily reaches the energy level zero then switches to \textsf{stop} directly. Hence, applying this strategy in the countdown game ensures that the sum of durations will be exactly equal to $c_0$ (recall that we start our $\AEL$ game by initializing the energy to $c_0$ then decrease it at every step by the duration chosen by $\playerOne$). Thus, this strategy is winning in the countdown game $\mathcal{C}$.

Second, consider the right-to-left direction. Assume that $\playerOne$ has a winning strategy in the countdown game $\mathcal{C}$. Playing this strategy in $G$ ensures to reach a state $v \in S_1$ with energy level exactly equal to zero. Thus a winning strategy for the $\AEL$ objective is to play the countdown strategy up to this point then to immediately take the edge $(v, \textsf{stop})$. Indeed, any consistent outcome will satisfy the lower bound on energy (as the energy will never go below zero), and it will have an average-energy equal to $t = 0$ (because the energy level when reaching \textsf{stop} will be zero).

This shows both directions of the claim and concludes our proof. Observe that this reduction is also true if we consider the $\AEinf$ variant of the average-energy.
\end{proof}

\paragraph{Memory requirements.} We close our study of two-player $\AEL$ games by discussing the memory needs. First note that we cannot provide upper bounds: if we had such bounds, we could derive a bound on the energy along any consistent play and reduce the $\AEL$ problem to an $\AELU$ one as discussed before, hence proving its decidability. Second, we already know by Thm.~\ref{thm:ael_one_memory} that pseudo-polynomial memory is necessary for $\playerOne$. Finally, we present a simple game (Fig.~\ref{fig:ael_two_memory}) where $\playerTwo$ needs to use memory in order to prevent $\playerOne$ from winning.

\begin{figure}[thb]
        \centering
\scalebox{0.9}{\begin{tikzpicture}[->,>=stealth',shorten >=1pt,auto,node
    distance=2.5cm,bend angle=45, scale=0.8, font=\normalsize]
    \tikzstyle{p1}=[draw,circle,text centered,minimum size=7mm,text width=4mm]
    \tikzstyle{p2}=[draw,rectangle,text centered,minimum size=7mm,text width=4mm]
    \node[p1]  (0) at (0, 0) {$s_{1}$};
    \node[p2]  (1)  at (3, 0) {$s_{2}$};
    \node[p1]  (2) at (6, 0) {$s_{3}$};
    
    \coordinate[shift={(-5mm,0mm)}] (init) at (0.west);
    \path
    (1) edge [loop left, out=130, in=50,looseness=2, distance=2cm] node [above] {$0$} (1)
    (2) edge [loop left, out=320, in=40,looseness=2, distance=2cm] node [right] {$-1$} (2)
    (init) edge (0)
    (0) edge node[above] {$1$} (1);
	\draw[->,>=latex] (1) to[out=00,in=180] node[above] {$-1$} (2);
	\draw[->,>=latex] (2) to[out=220,in=320] node[below] {$2$} (1);
      \end{tikzpicture}}
      \vspace{-2mm}
	\caption{Simple two-player $\AEL$ game witnessing the need for memory even for $\playerTwo$.}
	\label{fig:ael_two_memory}
\end{figure}

\begin{lemma}
\label{lem:ael_memory}
Pseudo-polynomial-memory strategies are necessary to win for $\playerOne$ in two-player $\AEL$ games. Memory is also required for $\playerTwo$ in such games.
\end{lemma}

\begin{proof}
We only have to prove that $\playerTwo$ needs memory in the game of Fig.~\ref{fig:ael_two_memory}. Consider the $\AEL$ objective for the average-energy threshold $t \coloneqq 1$ on this game. Assume that $\playerTwo$ is restricted to \textit{memoryless} strategies. Then, there are only two possible strategies for $\playerTwo$. If $\playerTwo$ always takes the self-loop $(s_{2}, s_{2})$, then the only consistent play is $s_1(s_2)^\omega$: it has $\AE$ equal to $1$, and satisfies the lower bound constraint on energy, thus $\playerOne$ wins. If $\playerTwo$ always takes $(s_2, s_3)$, then $\playerOne$ can win by producing the following play: $s_{1}s_2(s_3s_2s_3)^{\omega}$. It also has $\AE$ equal to $1$, and satisfies the energy constraint. Hence $\playerTwo$ cannot win this game with a memoryless strategy. Nonetheless, he has a winning strategy that uses memory. Let this strategy be the one that plays $(s_{2}, s_{3})$ once then chooses the self-loop $(s_2, s_2)$ forever. When this strategy is used by $\playerTwo$, $\playerOne$ has to pick $(s_3, s_2)$ in the first visit of $s_3$ otherwise he loses because the energy goes below zero. But if $\playerOne$ picks this edge, the unique outcome becomes $s_1s_2s_3(s_2)^{\omega}$, whose average-energy is $2 > t$, hence also losing for $\playerOne$. Thus, the defined strategy is winning for $\playerTwo$.
\end{proof}

\section{Conclusion}

We presented a thorough study of the \textit{average-energy} payoff. We showed that average-energy games belong to the same intriguing complexity class as mean-payoff, total-payoff and energy games and that they are similarly memoryless determined. We then solved average-energy games with lower- and upper-bounded energy: such a conjunction is motivated by previous case studies in the literature~\cite{CJLRR09}. Lastly, we provided preliminary results for the case of average-energy with a lower bound but no upper bound on the energy. Following the publication of~\cite{DBLP:journals/corr/BouyerMRLL15}, Larsen et al.~adressed a different problem in~\cite{DBLP:journals/corr/LarsenLZ15}: they proved that deciding if there exists a threshold $t \in \mathbb{Q}$ such that $\playerOne$ can win a two-player game for objective $\LBound(\initCredit \coloneqq 0)\,\cap\, \AvgEnergyLevel(t)$ can be done in doubly-exponential time. This is indeed equivalent to deciding if there exists an upper-bound $U \in \mathbb{N}$ such that $\playerOne$ can win for the objective $\LUBound(U, \initCredit \coloneqq 0)$, which is known to be in \textsf{2EXPTIME}~\cite{JLR13}. Unfortunately, this approach does not help in solving Problem~\ref{problem:ael}, where the threshold $t \in \mathbb{Q}$ for the average-energy is part of the input: solving two-player $\AEL$ games is still an open question.

We believe that the average-energy objective and its variations model relevant aspects of systems in practical applications as hinted by the aforementioned case study. Hence, we would like to extend our knowledge of this objective to more general models such as stochastic games, or games with multi-dimensional weights. Of course, the open questions regarding the $\AEL$ objective are intriguing. Finally, we would like to implement our techniques in synthesis tools and assess their applicability through proper case studies.

\bibliographystyle{plain}
\bibliography{biblio}

\end{document}